\documentclass[
aps,
pra,
10pt,
twocolumn,
superscriptaddress,
floatfix,
nofootinbib,
longbibliography
]{revtex4-2}

\usepackage{tikz}
\usepackage{graphicx}
\usepackage{physics}
\usetikzlibrary{trees, positioning, arrows.meta}
\usepackage{xcolor}
\usepackage[colorlinks=true,linkcolor=blue,urlcolor=magenta,citecolor=red]{hyperref} 
\usepackage{caption}
\usepackage{subcaption}
\usepackage{amsmath,amssymb,amsfonts}
\usepackage{amsthm}
\theoremstyle{plain}

\newtheorem{proposition}{Proposition}

\theoremstyle{definition}

\theoremstyle{remark}

\newtheorem{remark}{Remark}

\usepackage{bm}

\begin{document}

\title{A Quantum-Walk Representation of Color-Ordered MHV Scattering Amplitudes}

\author{Anirudh Verma}
\affiliation{Quantum Optics \& Quantum Information Laboratory, Dept. of Electronic Systems Engineering, Indian Institute of Science, Bengaluru 560012, India}

\author{C. M. Chandrashekar}
\affiliation{Quantum Optics \& Quantum Information Laboratory, Dept. of Electronic Systems Engineering, Indian Institute of Science, Bengaluru 560012, India}  

\begin{abstract}
\noindent We introduce a graph-theoretic framework for representing color-ordered maximally helicity violating (MHV) scattering amplitudes in quantum chromodynamics using coined quantum walks on permutation trees. Each root-to-terminal path corresponds to a distinct color ordering of the external gluons, while local transition amplitudes are assigned according to the spinor-product structure of the Parke--Taylor amplitudes. The walk evolves in coherent superpositions over permutation sectors, giving a dynamical picture of the underlying combinatorics. A quantum-channel formulation based on Kraus operators is also introduced to describe sector-resolved contributions, while a weighted collection operator coherently combines the terminal sectors at a common reference node. A quantum Fourier transform on the coin space is then employed to combine the encoded contributions into the corresponding color-decomposed amplitude. Together, these constructions establish a unified graph-based framework connecting permutation trees, quantum walks, and open quantum systems providing a framework for quantum algorithms to simulate scattering processes in quantum field theory. As an example, numerical results for low-point gluon amplitudes demonstrate that the proposed representation faithfully captures the characteristic Parke--Taylor structure and is consistent with analytical results.
\end{abstract}

\maketitle

\section{Introduction}
\noindent Scattering amplitudes play a central role in quantum field theory and high-energy physics, providing the fundamental quantities from which measurable cross sections and decay probabilities are computed \cite{Peskin1995,Weinberg1995,Schwartz2014}. In perturbative quantum chromodynamics (QCD), amplitudes involving many external gluons exhibit intricate combinatorial and algebraic structures arising from color degrees of freedom, gauge symmetry, and helicity dependence \cite{Ellis1996,ManganoParke1991}. Despite the large number of contributing Feynman diagrams, remarkable simplifications occur in specific helicity sectors, most notably in the maximally helicity violating (MHV) amplitudes discovered by Parke and Taylor \cite{ParkeTaylor1986}. These developments have led to modern approaches to scattering amplitudes based on spinor-helicity variables, recursion relations~\cite{Berends1988,BCFW2005}, twistor methods, and geometric formulations \cite{Xu1987,Dixon1996,Elvang2015,Witten2004,CHY2014}, providing compact analytic expressions for color-ordered gauge-theory amplitudes.\\
The color decomposition of gluon scattering amplitudes separates the full amplitude into color structures and kinematic partial amplitudes \cite{ManganoParke1991,Dixon1996},
\begin{equation}
M_n=\sum_{\sigma}C_\sigma A_\sigma,
\end{equation}
where $(C_\sigma)$ denotes the color factor associated with a permutation sector and $(A_\sigma)$ is the corresponding color-ordered partial amplitude. For MHV amplitudes, the Parke--Taylor formula expresses the kinematic contribution as
\begin{equation}
A(1^-,2^+,\dots,k^-,\dots,n^+)=
\frac{\langle1,k\rangle^4}
{\langle12\rangle
\langle23\rangle
\cdots
\langle n1\rangle},
\end{equation}
where the denominator consists of an ordered cyclic product of spinor inner products \cite{ParkeTaylor1986,Elvang2015}. The dependence on particle ordering gives the color decomposition an intrinsically combinatorial structure, with the number of color-ordered sectors growing factorially as the number of external particles increases.\\
Recent years have witnessed growing interest in applying ideas from quantum information to quantum field theory, including quantum algorithms for field-theoretic simulations, quantum representations of scattering processes, and quantum-computing approaches to high-energy physics \cite{Jordan2012,Klco2018,Bauer2023,Preskill2018}. More recently, highlighting the importance of quantum-information-based formulation of MHV amplitudes through quantum-circuit constructions and amplitude reconstruction algorithms~\cite{Bashore2025} has been proposed. In their proposal, quantum-circuit framework for MHV scattering amplitudes in which color, helicity, momentum, and permutation information are encoded into quantum registers and manipulated through structured unitary operations, leading to a reconstruction of the color-decomposed scattering amplitude.\\

In contrast, the present work formulates a graph-based representation in which color-ordered sectors are associated with paths on a directed permutation tree rather than encoding the scattering problem into quantum registers and circuit operations. In the graph-based model, each root-to-terminal path corresponds to a unique ordering of the external particles, while the topology of the graph reflects the combinatorial structure of the color decomposition. Within this framework, coined quantum-walk dynamics generate coherent superpositions over permutation sectors, and local transition amplitudes are chosen according to the spinor-product structure of the Parke--Taylor amplitudes. Consequently, the ordered denominator structure of the color-ordered amplitudes is represented through successive transitions on the permutation graph. Furthermore, the present framework incorporates a quantum-channel formulation for sector-resolved amplitude extraction and reconstructs the color-decomposed amplitudes using a weighted collection operator together with a quantum Fourier transform. These features provide a complementary graph-based and dynamical perspective to the circuit-based formulation of Ref.~\cite{Bashore2025}. A detailed table with comparison between the present graph-based construction and the quantum-circuit formulation of Ref.~\cite{Bashore2025} is presented in the appendix (Table~\ref{tab:comparison}).\\

The quantum walk formalism is an highly successful framework for modeling quantum quantum dynamics, such as neutrino oscillations \cite{Mallick2017,Sahu2023}, Dirac equation different dynamics in both the low and high energy regime  \cite{CM2013,Mallick2016, Kumar2018,Mallick2019}, and for developing a wide range of quantum algorithms to model dynamics on complex quantum networks, graph exploration \cite{Aharonov1993,Meyer1996,Ambainis2003,Kempe2003,Venegas2012,Portugal2013}.
The experimental demonstration of quantum walk to simulate the Dirac equation \cite{alderete2020} and as a  
powerful primitive for quantum computation \cite{Sengupta2025} further strengthens the practical relevance of the quantum walk framework for modeling scattering behavior in field theory dynamics.  Since quantum walk evolution is naturally defined on graph structures, they are well suited for describing systems in which combinatorial organization plays a fundamental role. Motivated by these developments and to provide an algorithmic framework which can be implemented on a quantum processors in near future, we present a coined quantum walk on the permutation tree to provide a dynamical representation of the color-ordering structure of MHV amplitudes.\\

To analyze the resulting graph dynamics, we introduce a quantum-channel formulation in which Kraus operators act on the terminal permutation sectors to extract sector-resolved quantities proportional to the squared color-ordered amplitudes $( |A_\sigma|^2 )$ \cite{Kraus1983,Breuer2002,Nielsen2010}. This channel description complements the underlying unitary quantum-walk evolution and establishes a connection with graph-based quantum dynamics and open quantum walks \cite{Attal2012}. We further introduce a weighted collection operator that coherently combines the terminal sectors at a common reference node while preserving their relative weights. A quantum Fourier transform acting on the coin space is then used to combine the encoded contributions into the corresponding color-decomposed scattering amplitude. This establishes a neat connection between scattering amplitudes, permutation trees, coined quantum walks, and quantum channels. Rather than viewing the color decomposition solely as an algebraic sum over permutations, it also provides a graph-theoretic representation in which the ordering structure is encoded in the geometry of the permutation tree and explored through quantum-walk dynamics. Numerical examples for low-point gluon amplitudes demonstrate that the resulting construction faithfully reproduces the characteristic Parke--Taylor structure and agrees with analytical calculations.\\

The remainder of the paper is organized as follows.In Sec.~\ref{sec:QCDA_PTS} we review the color decomposition of tree-level QCD amplitudes,the Parke--Taylor MHV amplitudes, and introduce the permutation-tree representation of the color-ordering structure. In Sec.~\ref{sec:qw_construction} we develop the coined quantum-walk framework by constructing the directed permutation graph together with globally consistent edge labeling, local coin operators, and a unitary shift operator. Section~\ref{pkt} demonstrates how the resulting quantum walk generates coherent superpositions over color-ordered permutation sectors and establishes the connection between the accumulated transition amplitudes and the Parke--Taylor denominator structure.
In Sec.~\ref{sec:CAR} we formulate the quantum-channel description, introduce the weighted collection operator and quantum Fourier transform for coherent amplitude reconstruction, and present the complete reconstruction algorithm. Finally, Sec.~\ref{sec:NV} presents numerical demonstrations for representative low-point gluon scattering amplitudes together with comparisons to the corresponding analytical results.

\section{QCD Amplitudes, Parke--Taylor Structure and Permutation Trees} \label{sec:QCDA_PTS}
\noindent In the high-energy regime,\(Q^2 \gg \Lambda_{\mathrm{QCD}}^2\), the strong coupling becomes sufficiently small for perturbative calculations to be applicable~\cite{Peskin1995,Ellis1996}. A generic tree-level QCD amplitude can be decomposed into color and kinematic contributions,
\begin{equation}
M_n=\sum_{\sigma}C_\sigma A_\sigma,
\end{equation} where \(C_\sigma\) are color factors and \(A_\sigma\) are color-ordered partial amplitudes. For pure gluon scattering, a convenient trace decomposition is given by~\cite{ManganoParke1991,Dixon1996,DelDuca2000}

\begin{equation}
M_n
=
g^{n-2}
\sum_{\sigma\in S_{n-1}}
\mathrm{Tr}
\left(
T^{a_1}
T^{a_{\sigma(2)}}
\cdots
T^{a_{\sigma(n)}}
\right)
A(1,\sigma(2,\ldots,n)).
\end{equation} This decomposition separates the group-theoretic color structure from the kinematic information while explicitly organizing the scattering amplitude according to permutations of the external particles. Consequently, the combinatorial structure of gluon scattering is naturally encoded in the set of color-ordered permutation sectors.\\
Among the different helicity configurations, the dominant tree-level contributions arise from the maximally helicity violating (MHV) sector~\cite{ParkeTaylor1986,ManganoParke1991}. For an \(n\)-gluon process with two negative-helicity gluons, the corresponding Parke--Taylor amplitude takes the compact form
\begin{equation}
A(1^-,2^+,\ldots,k^-,\ldots,n^+)
=
\frac{\langle1k\rangle^4}
{\langle12\rangle
\langle23\rangle
\cdots
\langle n1\rangle},
\end{equation} where \(\langle ij\rangle\) denotes the spinor-helicity inner product~\cite{Xu1987,Dixon1996,Elvang2015}. Throughout this work, the first and the \(k\)-th gluons carry negative helicity while the remaining gluons have positive helicity.\\
For a fixed color ordering \((1,\sigma(2),\ldots,\sigma(n))\), the corresponding color-ordered amplitude can be written as
\begin{equation}
A_\sigma
=
\frac{\langle1k\rangle^4}
{\langle1\,\sigma(2)\rangle
\langle\sigma(2)\,\sigma(3)\rangle
\cdots
\langle\sigma(n)\,1\rangle}.
\label{eq:Asigma}
\end{equation} The denominator consists of an ordered cyclic product of spinor inner products determined entirely by the chosen permutation. This ordered structure suggests a natural combinatorial representation in which each permutation is associated with a sequence of local transitions, while the complete denominator is obtained by successively traversing the corresponding ordering.\\
Motivated by this observation, we represent the set of color-ordered sectors using a permutation tree. Starting from a root node, new particles are recursively appended to generate all admissible partial orderings, producing a directed tree whose terminal nodes correspond to complete permutations. Each root-to-terminal path therefore uniquely represents a color-ordered sector of the scattering amplitude. Within this representation, the factorial growth of the permutation sectors is encoded directly in the graph topology rather than being treated as an explicit sum over permutations. The permutation tree therefore provides a natural graph on which quantum-walk dynamics can be defined, allowing coherent superpositions over different color orderings while preserving the underlying combinatorial structure of the amplitude.\\

\section{Quantum Walk Construction}
\label{sec:qw_construction}

To formulate a quantum walk~\cite{Aharonov1993,Kempe2003,Ambainis2003} that generates all color-ordered
sectors of the scattering amplitude, we first establish the
combinatorial structure of the permutation tree as a directed
graph, fix a global labeling of its edges, and use that
labeling to define both the coin and shift operators in a
mutually consistent way.\\
Before introducing the formal graph-theoretic definitions, it is useful to describe the construction intuitively.For an $n$-gluon color-ordered amplitude, particle $1$ is kept fixed while the remaining particles $\{2,\ldots,n\}$ are arranged in every possible order. Rather than generating all complete permutations directly, we build them recursively.\\
The construction begins at the root node $(1)$, representing the trivial partial ordering containing only the fixed particle. At each step, one particle that has not yet appeared is appended to the current sequence, thereby creating a child node. Repeating this procedure generates a directed tree whose vertices represent partial permutations and whose terminal nodes correspond to complete permutations of the remaining particles.\\
For example, in the four-gluon case the root node $(1)$ has three children,
\[
(1,2),\qquad (1,3), \qquad (1,4),
\]
since particles $2$, $3$, and $4$ are all initially available. From the node $(1,3)$, particle $3$ has already been used, so only particles $2$ and $4$ remain available, producing the children
\[
(1,3,2),\qquad (1,3,4).
\]
Continuing recursively eventually generates all $(n-1)!$ complete color orderings.\\
This recursive construction naturally defines a graph on which a coined quantum walk can be performed. The role of the coin operator is to distribute quantum amplitude among the admissible outgoing edges according to the local spinor-product weights, while the shift operator propagates each coin component to the corresponding child node. Consequently, every root-to-terminal path accumulates the sequence of transition amplitudes associated with one particular color ordering.

\begin{figure}[h]
\centering
\includegraphics[
width=0.5\textwidth,
keepaspectratio
]{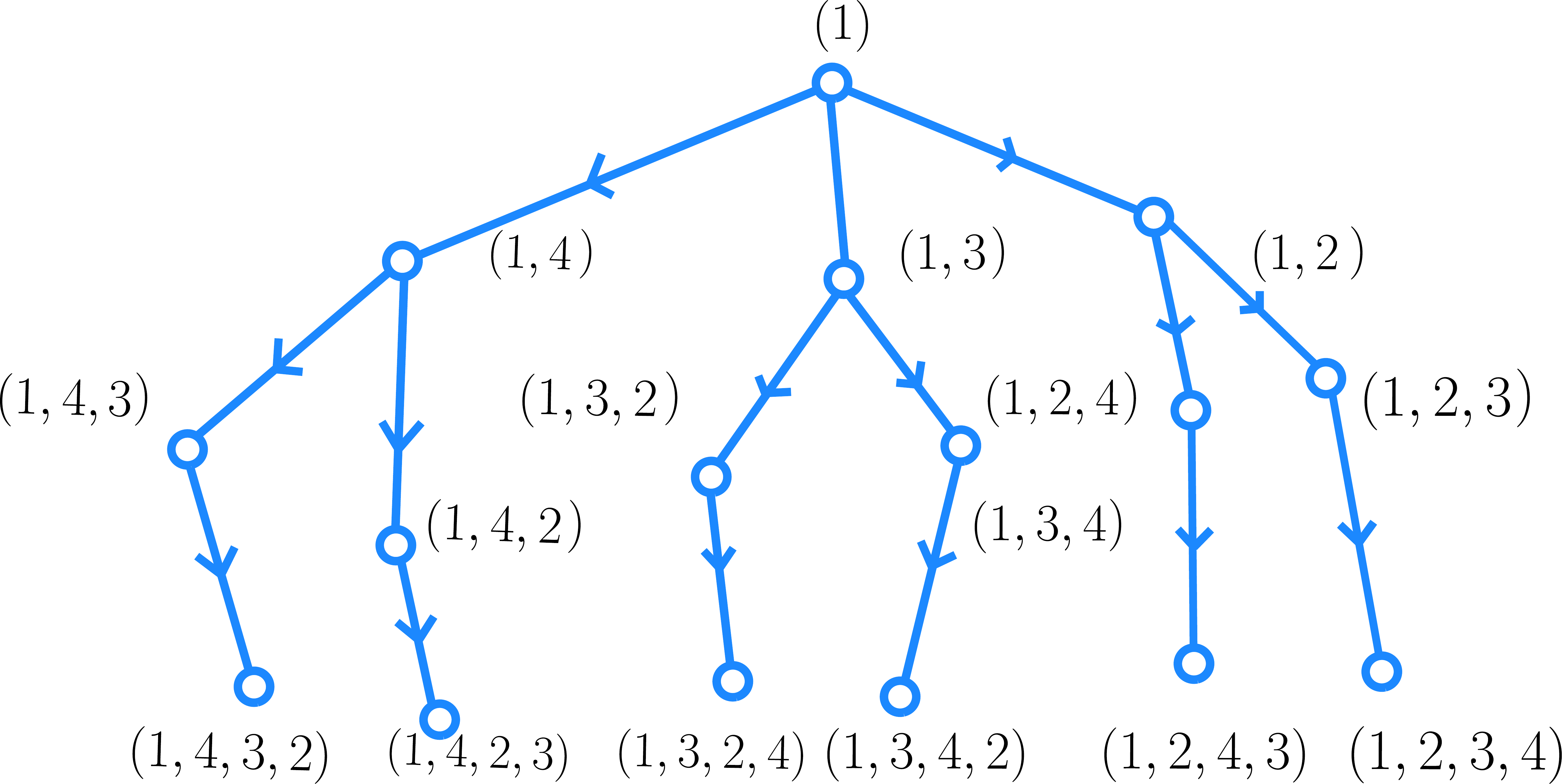}
\caption{
Directed permutation tree for the four-gluon case. Each node represents
a partial permutation
$v=(1,a_2,\ldots,a_m)$,
while every directed edge corresponds to the append operation
$v\rightarrow(v,a)$ for
$a\in R(v)$.
Every root-to-terminal path uniquely determines one complete
color ordering of the external particles.
}
\label{fig:permutation_tree}
\end{figure}

The recursive construction described above is illustrated in Fig.~\ref{fig:permutation_tree} for the four-gluon case. Every node represents a partial permutation beginning with the fixed particle $1$, while every directed edge corresponds to appending one additional particle that has not yet appeared in the sequence. Since each particle is appended exactly once, every root-to-terminal path uniquely determines one complete color ordering of the external particles.

\subsection{Directed graph and edge labeling}
\label{ssec:graph}

\paragraph{Global particle-to-coin map.}
Fix the set of external particles $\{1,2,\ldots,n\}$,
where particle~$1$ is the reference particle.
Define the \emph{global particle-to-coin map}
\begin{equation}
  c:\{2,\ldots,n\}\to\{0,\ldots,n-2\},
  \qquad
  c(j)=j-2.
  \label{eq:global_coin_map}
\end{equation}
This map is a bijection.
Throughout this work, coin value $k$ corresponds to
particle $k+2$; this assignment is fixed and independent
of the position of the walker on the graph.

\paragraph{Permutation tree as a directed graph.}
Define the directed graph $G=(V,E)$ as follows.

\begin{itemize}
  \item \textbf{Vertex set.}
  \begin{equation}
  \begin{split}
    V \;=\; & \bigl\{\,(1,a_2,\ldots,a_{m+1})
               \;\big|\;
               0\le m\le n-1,\;
               a_i\in\{2,\ldots,n\},\; \\
               & a_i\neq a_j \text{ for } i\neq j
            \bigr\}.
    \label{eq:vertex_set}
    \end{split}
  \end{equation}
  The unique node with $m=0$ is the \emph{root} $r=(1)$.
  Nodes with $m=n-1$ are \emph{terminal}; denote their set
  $\mathcal{T}\subset V$.  We have $|\mathcal{T}|=(n-1)!$ and
  $|V|=\sum_{k=0}^{n-1}(n-1)!/k!\,$.

  \item \textbf{Edge set.}
  There is a directed edge $v\to u$ whenever
  $u=(v,a)$ for some $a\in\{2,\ldots,n\}\setminus\{a_2,\ldots,a_{m+1}\}$.
  Write $\mathrm{Out}(v)$ for the set of out-neighbors of $v$.
  Non-terminal nodes satisfy $|\mathrm{Out}(v)|=n-1-m\ge 1$;
  terminal nodes satisfy $|\mathrm{Out}(t)|=0$.
\end{itemize}

\paragraph{Remaining-particle and placed-particle sets.}
For $v=(1,a_2,\ldots,a_{m+1})\in V$ define
\begin{equation}
  R(v) = \{2,\ldots,n\}\setminus\{a_2,\ldots,a_{m+1}\},
  \;\;
  P(v) = \{a_2,\ldots,a_{m+1}\},
  \label{eq:remaining}
\end{equation}
so that $\mathrm{Out}(v)=\{(v,r)\mid r\in R(v)\}$,
$|R(v)|=n-1-m$, and $R(v)\cup P(v)=\{2,\ldots,n\}$.

\paragraph{Edge-label function (global).}
\label{para:edge_label}
Using the global map~\eqref{eq:global_coin_map}, define the
\emph{edge-label function}
\begin{equation}
  \lambda\bigl(v,(v,r)\bigr) \;=\; c(r) \;=\; r-2
  \qquad
  \text{for every } r\in R(v).
  \label{eq:edge_label}
\end{equation}
Thus each outgoing edge of $v$ carries the \emph{global} coin
label of the particle it appends.
Unlike a local labeling that renumbers remaining particles at
each node, the label assigned to a particle is the same at every
node where that particle is available.
The outgoing edge labels at $v$ therefore form the set
$\{c(r):r\in R(v)\}\subseteq\{0,\ldots,n-2\}$,
which in general is a proper subset of $\{0,\ldots,n-2\}$:
labels corresponding to already-placed particles are absent.

\paragraph{Incoming-coin function.}
\label{para:kappa}
For every non-root node
$v=(1,a_2,\ldots,a_{m+1})\in V\setminus\{r\}$,
define the \emph{incoming coin label} as
\begin{equation}
  \kappa(v) \;=\; c(a_{m+1}) \;=\; a_{m+1}-2,
  \label{eq:kappa}
\end{equation}
i.e.\ the global coin label of the \emph{last appended particle}.
For the root set $\kappa(r)=0$ by convention (the walk is
initialized with coin state $|0\rangle$).

The function $\kappa$ is well-defined because every non-root node
has a unique last element.
It satisfies $\kappa(v)=\lambda\!\bigl(\mathrm{par}(v),v\bigr)$
by construction, since the edge from $\mathrm{par}(v)$ to $v$
appends particle $a_{m+1}$ and thus carries label $c(a_{m+1})$.

\paragraph{Example ($n=4$).}
The global map is $c(2)=0$, $c(3)=1$, $c(4)=2$.
At the root: $R((1))=\{2,3,4\}$, giving
$\lambda\bigl((1),(1,2)\bigr)=0$,
$\lambda\bigl((1),(1,3)\bigr)=1$,
$\lambda\bigl((1),(1,4)\bigr)=2$.
All three coin labels are present.

At node $(1,3)$: $R((1,3))=\{2,4\}$, $P((1,3))=\{3\}$, giving
$\lambda\bigl((1,3),(1,3,2)\bigr)=0$ and
$\lambda\bigl((1,3),(1,3,4)\bigr)=2$.
The label set is $\{0,2\}$; label $1$ (corresponding to
particle~$3$, already placed) is absent.

Incoming coin labels:
$\kappa\bigl((1,2)\bigr)=0$,
$\kappa\bigl((1,3)\bigr)=1$,
$\kappa\bigl((1,4)\bigr)=2$.

The directed permutation tree corresponding to the four-gluon case is
illustrated in Fig.~\ref{fig:permutation_tree}. Each node represents a
partial permutation beginning with the fixed particle $1$, while each
directed edge corresponds to appending one additional particle according
to the graph construction described above. Every root-to-terminal path
therefore represents one unique color ordering contributing to the
scattering amplitude.

\subsection{Hilbert space}
\label{ssec:hilbert}

The coined quantum walk can be described on the composite Hilbert Space consisting of 
\begin{equation}
  \mathcal{H} \;=\; \mathcal{H}_{\mathrm{pos}}
                   \otimes
                   \mathcal{H}_{\mathrm{coin}},
  \label{eq:hilbert}
\end{equation}
where the \emph{position space} is
$\mathcal{H}_{\mathrm{pos}}=\mathrm{span}\{|v\rangle:v\in V\}$
and the \emph{coin space} is
$\mathcal{H}_{\mathrm{coin}}=\mathrm{span}\{|k\rangle:k=0,\ldots,d-1\}$
with
\begin{equation}
  d \;=\; n-1.
  \label{eq:coin_dim}
\end{equation}
A basis state $|v,k\rangle=|v\rangle\otimes|k\rangle$ represents
the walker at node $v$ with coin value $k$.
For $n=4$ we have $|V|=16$, $d=3$, and
$\dim\mathcal{H}=48$.

\subsection{Coin operator}
\label{ssec:coin}

\paragraph{Local transition amplitudes.}
For each non-terminal node $v$ and each out-neighbor
$(v,r)\in\mathrm{Out}(v)$ with $r\in R(v)$, define the
\emph{raw transition weight}
\begin{equation}
  w(v,r) \;=\; \frac{1}{\langle\ell(v)\;r\rangle},
  \label{eq:raw_weight}
\end{equation}
where $\ell(v)$ is the last particle label in the sequence $v$.
Define the local normalization constant
\begin{equation}
  \alpha_v \;=\;
  \sqrt{\sum_{r\in R(v)}|w(v,r)|^2}\,,
  \label{eq:alpha}
\end{equation}
and the normalized amplitude
\begin{equation}
  \tilde{v}_{v,r} \;=\; \frac{w(v,r)}{\alpha_v}.
  \label{eq:norm_amp}
\end{equation}
By construction $\sum_{r\in R(v)}|\tilde{v}_{v,r}|^2=1$.

\paragraph{Local coin operator.}
\label{para:local_coin}
For each non-terminal $v$, define a $d\times d$ unitary matrix
$C_v\in U(d)$ satisfying the following column condition.

Let $\kappa(v)$ be the incoming coin label
(defined in \eqref{eq:kappa}).  Require that the
$\kappa(v)$-th column of $C_v$ encodes the normalized
transition amplitudes, placed at the row corresponding to
each remaining particle's global coin label:
\begin{equation}
  (C_v)_{c(r),\,\kappa(v)} \;=\; \tilde{v}_{v,r}
  \quad\text{for every } r\in R(v),
  \label{eq:coin_column}
\end{equation}
with $(C_v)_{k,\,\kappa(v)}=0$ for every
$k\in\{c(a):a\in P(v)\}$ (rows corresponding to
already-placed particles are zero).  The remaining
$d-1$ columns of $C_v$ are completed to a unitary matrix by
Gram--Schmidt orthonormalization.

\begin{remark}[Coin--shift consistency]
When the walker arrives at $v$ carrying coin state
$|\kappa(v)\rangle$, the coin operator $C_v$ maps
$|\kappa(v)\rangle$ to a superposition
$\sum_{r\in R(v)}\tilde{v}_{v,r}\,|c(r)\rangle$
over exactly those coin values whose corresponding particles
are still available.  The shift operator (defined below) then
routes each component $|c(r)\rangle$ to the child $(v,r)$.
In particular, the coin operator places zero amplitude in
$|\kappa(v)\rangle$ itself (since particle $\kappa(v)+2$ is
already the last element of $v$ and hence in $P(v)$), so the
walker is always scattered away from the incoming coin direction.
\end{remark}

For terminal nodes $t\in\mathcal{T}$, no forward transition is
available; set $C_t=\mathbf{1}_d$ (identity).

\paragraph{Global coin operator.}
The full coin operator acts locally on each node:
\begin{equation}
  C \;=\; \bigoplus_{v\in V} C_v \;\in\; U(d|V|).
  \label{eq:global_coin}
\end{equation}
Its action on a basis state is
$C|v,k\rangle=\sum_{j=0}^{d-1}(C_v)_{jk}\,|v,j\rangle$.

\subsection{Shift operator}
\label{ssec:shift}

The shift operator $S$ implements the conditional propagation
of the walker along the directed edges of $G$~\cite{Kempe2003,Portugal2013}.
We define it as a permutation
$\tau:V\times\{0,\ldots,d-1\}\to V\times\{0,\ldots,d-1\}$
and set $S|v,k\rangle=|\tau(v,k)\rangle$.

\paragraph{Partition of basis states.}
Fix a coin sector $k\in\{0,\ldots,d-1\}$ and let $p=k+2$
denote the corresponding particle.
Partition $V\times\{k\}$ into three disjoint sets:
\begin{align}
  F_k &= \bigl\{(v,k)\;\big|\;
          v\notin\mathcal{T},\;
          p\in R(v)
         \bigr\}, \label{eq:Fk}\\[4pt]
  L_k &= \bigl\{(v,k)\;\big|\;
          v\notin\mathcal{T},\;
          p\in P(v),\;
          k\neq\kappa(v)
         \bigr\}, \label{eq:Lk}\\[4pt]
  D_k &= \bigl\{(v,k)\;\big|\;
          v\in\mathcal{T}
         \bigr\}
         \;\cup\;
         \bigl\{(v,k)\;\big|\;
          v\notin\mathcal{T},\;
          k=\kappa(v),\;
          v\neq r
         \bigr\}. \label{eq:Dk}
\end{align}
$F_k$ (\emph{forward}) contains non-terminal states where
particle $p$ has not yet been placed.
$L_k$ (\emph{loop}) contains non-terminal states where particle
$p$ has been placed but $k$ is not the incoming coin label.
$D_k$ (\emph{displaced}) contains all terminal states together
with non-root states where $k$ equals the incoming coin label
(i.e.\ particle $p$ was the last particle appended).

\begin{remark}[Physical interpretation of $D_k$]
\label{rem:Dk}
A non-terminal state $(v,\kappa(v))$ falls in $D_k$ rather than
$L_k$ because assigning it a self-loop would create a
non-injectivity: the forward rule already maps
$|\mathrm{par}(v),\kappa(v)\rangle\to|v,\kappa(v)\rangle$,
so a self-loop at $(v,\kappa(v))$ would give
$|v,\kappa(v)\rangle$ two pre-images.
Placing these states in $D_k$ (alongside terminal states) and
mapping them to otherwise-uncovered targets eliminates this
conflict and makes $S$ a bijection.
During the first $n-1$ walk steps, these states are never populated
because the coin operator always scatters the walker away from
the incoming coin direction
($(C_v)_{\kappa(v),\kappa(v)}=0$ ; see~\eqref{eq:coin_column}),
so they are dynamically equivalent to self-loops.
This explains why the graph figures depict self-loop arrows at
these positions.
\end{remark}

\paragraph{Shift rules.}
The permutation $\tau$ is defined by:
\begin{equation}
  \boxed{
  \tau(v,k) \;=\;
  \begin{cases}
    \bigl((v,\,k{+}2),\;k\bigr)
      & \text{if } (v,k)\in F_k
        \;\;\text{(append} \\ & \text{ particle } k{+}2\text{)},\\[4pt]
    (v,\,k)
      & \text{if } (v,k)\in L_k
        \;\;\text{(self-loop)},\\[4pt]
    \bigl(\varphi_k(v),\;k\bigr)
      & \text{if } (v,k)\in D_k
        \;\;\text{(displaced;} \\ &\text{ reassignment)},
  \end{cases}
  }
  \label{eq:shift_complete}
\end{equation}
where $\varphi_k$ is a bijection from the source set of $D_k$
to the set $\mathcal{U}_k$ of uncovered targets defined below.

For the four-gluon example, the partition
\begin{equation}
V\times\{k\}=F_k\cup L_k\cup D_k
\end{equation}
can be evaluated explicitly for every node of the permutation tree.
The complete classification of all basis states is provided in
Appendix~\ref{app:shift_classification}.
This explicit example illustrates how every basis state belongs to
exactly one of the three subsets used in the construction of the
unitary shift operator and serves as a direct verification of the
algorithm presented above.

\paragraph{Uncovered targets.}
A state $(w,k)$ is an \emph{uncovered target} if it is neither
a forward-rule image nor a self-loop fixed point:
\begin{equation}
  \mathcal{U}_k =
  \bigl\{(w,k)\;\big|\;
    \kappa(w)\neq k\text{ or }w=r
  \bigr\}
  \setminus L_k.
  \label{eq:uncovered}
\end{equation}
A state $(w,k)$ is a forward-rule image if and only if
$\kappa(w)=k$ (the last element of $w$ is particle $k{+}2$)
and $w\neq r$.
A state in $L_k$ is its own image.
The uncovered targets are everything else.

\begin{proposition}[Bijectivity]
\label{prop:S_unitary}
$|D_k|=|\mathcal{U}_k|$ for every $k$, so a bijection
$\varphi_k:D_k\to\mathcal{U}_k$ exists.
The resulting map $\tau$ is a permutation on
$V\times\{0,\ldots,d{-}1\}$, and $S$ is unitary.
\end{proposition}

\begin{proof}
The three source sets $F_k$, $L_k$, $D_k$ partition
$V\times\{k\}$.
The forward rule is injective (each child has a unique parent
per coin label), the self-loop rule is the identity on $L_k$,
and $\varphi_k$ is a bijection onto $\mathcal{U}_k$ by
definition.
It remains to show the three image sets are disjoint.
\begin{itemize}
\item Forward images are states $(u,k)$ with $\kappa(u)=k$ and
$u\neq r$.
\item Self-loop images are the states in $L_k$, which satisfy
$k\neq\kappa(v)$.
\item Uncovered targets satisfy $\kappa(w)\neq k$ (or $w=r$)
and $(w,k)\notin L_k$.
\end{itemize}
The condition $\kappa(u)=k$ versus $k\neq\kappa(v)$ versus
the complement ensures disjointness.
Since the three image sets are disjoint and each is injectively
covered, $\tau$ is a bijection on $V\times\{k\}$.
This holds for every $k$, so $\tau$ is a permutation on the full
basis and $S$ is unitary.
\end{proof}

\paragraph{Canonical choice of $\varphi_k$.}
To make $S$ fully explicit, we index the elements of $D_k$ and
$\mathcal{U}_k$ in lexicographic order of the node sequences
and set $\varphi_k(i\text{-th source})=i\text{-th target}$.

\paragraph{Algorithm for constructing the shift operator.}

For every basis state $|v,k\rangle$:

\begin{enumerate}
\item Compute the particle label
      \[
      r=c^{-1}(k).
      \]

\item If $r\in R(v)$, move the walker to the child obtained by appending
      particle $r$,
      \[
      S|v,k\rangle=|(v,r),k\rangle.
      \]

\item Otherwise, apply the bijection $\phi_k$ to complete the permutation,
      \[
      S|v,k\rangle=|\phi_k(v),k\rangle.
      \]
\end{enumerate}
The shift operator therefore acts by interpreting the coin label through the
fixed particle-to-coin map. Whenever the corresponding particle has not yet
appeared in the current partial permutation, the walker propagates along the
unique child obtained by appending that particle. Otherwise, the state belongs
to the displaced subspace and is mapped through the bijection $\phi_k$,
ensuring that the complete shift operator is a permutation on the Hilbert
space and hence unitary.

\begin{figure}[t]
\centering
\includegraphics[width=0.78\linewidth]{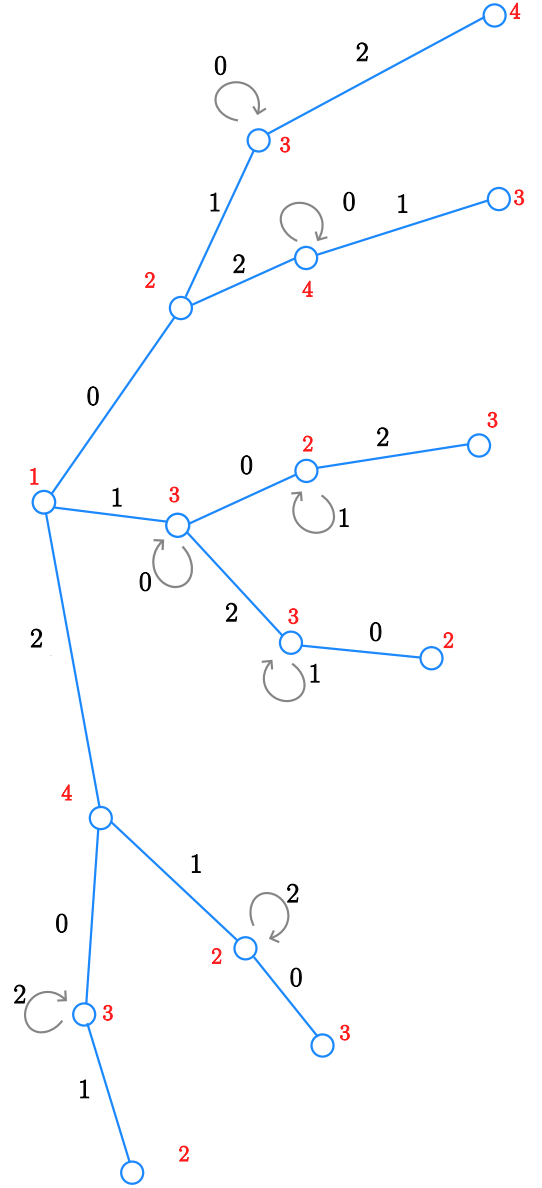}
\caption{
Schematic illustration of the local construction of the shift operator
for the four-gluon permutation tree.
Blue edges represent the forward transitions associated with the subset
$F_k$, while gray self-loops correspond to basis states belonging to the
loop subset $L_k$.
The red labels indicate the particle appended during each forward
transition, whereas the black labels denote the associated coin values.
For clarity, the displaced subset $D_k$ is not shown explicitly.
Its action is completed by choosing an arbitrary bijection onto the
remaining unassigned basis states of
$V\times\{0,\ldots,d-1\}$, thereby extending the local construction to a
permutation of the complete Hilbert-space basis and ensuring the
unitarity of the shift operator.
}
\label{fig:shift_construction}
\end{figure}
Figure~\ref{fig:shift_construction} illustrates only the local action of
the shift operator on the subsets $F_k$ and $L_k$. The forward
transitions and self-loops uniquely follow from the graph structure and
the partition defined in Eqs.~(\ref{eq:Fk})--(\ref{eq:Dk}). The action on
the displaced subset $D_k$ is intentionally omitted since it is not
unique. As given in Proposition~\ref{prop:S_unitary}, any bijection
between the displaced states and the remaining unassigned basis states
produces a valid permutation of
$V\times\{0,\ldots,d-1\}$ and therefore defines an equivalent unitary
shift operator.

\begin{remark}
The physics of the amplitude reconstruction is independent of
the choice of $\varphi_k$.  During the first $n{-}1$ walk steps,
every state in $D_k$ at a non-terminal node has zero amplitude
(the coin operator places no amplitude in the incoming coin
direction), and terminal states are reached only at step $n{-}1$.
Since the reassignment rule first acts at step $n$ and beyond,
the reconstruction is $\varphi_k$-independent.
\end{remark}

\begin{table}[h]
\centering
\caption{Summary of the notation used in the permutation-tree construction.}
\begin{tabular}{ll}
\hline
Symbol & Meaning \\
\hline
$V$ & Set of vertices (partial permutations) \\
$E$ & Set of directed edges \\
$r=(1)$ & Root node \\
$T$ & Set of terminal nodes (complete permutations) \\
$R(v)$ & Remaining particles that have not yet appeared in $v$ \\
$P(v)$ & Particles already contained in $v$ \\
$\mathrm{Out}(v)$ & Set of children of $v$ \\
$c(j)$ & Global particle-to-coin map \\
$\lambda(v,u)$ & Coin label assigned to edge $v\rightarrow u$ \\
$\kappa(v)$ & Incoming coin label of node $v$ \\
$F_k$ & Forward states for coin value $k$ \\
$L_k$ & Loop states for coin value $k$ \\
$D_k$ & Displaced states used to complete the shift permutation \\
$\phi_k$ & Bijection completing the shift operator \\
\hline
\end{tabular}
\label{tab:notation}
\end{table}

\subsection{Walk operator and initial state}
\label{ssec:walk_op}

The walker is initialized at the root with coin state $|0\rangle$:
\begin{equation}
  |\psi_0\rangle \;=\; |(1)\rangle\otimes|0\rangle.
  \label{eq:initial_state}
\end{equation}
The coined quantum-walk operator is
\begin{equation}
  U \;=\; S\,C,
  \label{eq:walk_op}
\end{equation}
where $C$ is the global coin operator \eqref{eq:global_coin}
and $S$ is the shift \eqref{eq:shift_complete}.
After $t$ steps the state of the walker is
\begin{equation}
  |\psi_t\rangle \;=\; U^t|\psi_0\rangle \;=\; (SC)^t|\psi_0\rangle.
  \label{eq:walk_state}
\end{equation}
Since both $C$ and $S$ are unitary, $U$ is unitary for all $t$.
After $n-1$ steps --- the depth of the permutation tree --- the
walker reaches the terminal permutation sectors,
\begin{equation}
  |\psi_{\mathrm{fin}}\rangle
  \;=\; U^{n-1}|\psi_0\rangle
  \;=\; \sum_{\sigma\in S_{n-1}} \psi_\sigma\,|\sigma\rangle
        \otimes|\kappa_\sigma\rangle,
  \label{eq:final_state}
\end{equation}
where each terminal node $\sigma=(1,\sigma(2),\ldots,\sigma(n))$
uniquely labels a color-ordered permutation sector, and
$\kappa_\sigma=c(\sigma(n))=\sigma(n)-2$ is the global coin label
of the last appended particle.
The derivation of the explicit form of the path amplitudes
$\psi_\sigma$ and their correspondence to the Parke--Taylor
denominator is given in Appendix~\ref{coh_evo}.

\section{Amplitude Reconstruction}\label{pkt} The coined quantum walk on the permutation tree
reconstructs the color-ordered Parke--Taylor amplitudes.
\subsection{Reconstruction of Parke--Taylor Amplitudes}
\noindent Starting from the initial state $\ket{\psi_0} = \ket{(1)}\otimes \ket{0}$ the walker evolves under repeated application of the coined quantum-walk operator $U=SC$ where \(C\) denotes the local coin operator and \(S\) is the conditional shift operator. After \(n-1\) steps, corresponding to the depth of the permutation tree, the walker reaches the terminal permutation sectors,$ \ket{\psi_{\mathrm{fin}}} = (SC)^{n-1} \ket{\psi_0}$. As shown in Appendix \ref{coh_evo}, the repeated action of the local coin and  shift operators generates a coherent superposition over every root-to-terminal path, 
\begin{equation}
\ket{\psi_{\mathrm{fin}}} =\sum_{\sigma\in S_{n-1}} \psi_\sigma \,\ket{\sigma} \otimes \ket{k_\sigma},\label{eq:psi_final_main}
\end{equation} where each terminal node uniquely labels a permutation sector.\\
The path amplitudes are derived in Appendix~E by tracking the
product of coin-matrix entries accumulated along each
root-to-terminal path.  The result is
\begin{equation}
  \psi_\sigma
  \;=\;
  \frac{1}{\displaystyle\prod_{r=1}^{n-1}\alpha_{\sigma(r)}}
  \cdot
  \frac{1}{\langle 1\,\sigma(2)\rangle
            \langle\sigma(2)\,\sigma(3)\rangle
            \cdots
            \langle\sigma(n)\,1\rangle},
  \label{eq:psi_sigma_explicit}
\end{equation}
where $\alpha_{\sigma(r)}$ are the local normalization constants
defined in~\eqref{eq:alpha}.
Comparing with the Parke--Taylor amplitude~\eqref{eq:Asigma},
one obtains the explicit proportionality
\begin{equation}
  \psi_\sigma
  \;=\;
  \frac{\langle\sigma(n)\,1\rangle}
       {\langle 1\tilde{k}\rangle^4\;
        \displaystyle\prod_{r=1}^{n-1}\alpha_{\sigma(r)}}
  \;A_\sigma,
  \label{eq:psi_Asigma}
\end{equation}
so that $\psi_\sigma\propto A_\sigma$, with the proportionality
constant fully determined by known kinematic factors.
The detailed recursive derivation is presented in Appendix~\ref{coh_evo}.

\subsection{Open Quantum-System Formulation}
To extract the contribution of each terminal permutation sector,
we construct the density operator
\begin{equation}
  \rho \;=\; |\psi_{\mathrm{fin}}\rangle\langle\psi_{\mathrm{fin}}|
  \label{eq:density_op}
\end{equation}
and introduce a collection of Kraus operators~\cite{Kraus1983,%
Breuer2002,Nielsen2010} acting on the terminal
nodes, $K_\sigma = c_\sigma|\sigma\rangle\langle\sigma|\otimes
I_{\mathrm{coin}}$, where the coefficients $c_\sigma$ are chosen
such that $\sum_\sigma K_\sigma^\dagger K_\sigma\le I$.
The explicit form of $c_\sigma$ is given in Appendix~\ref{KC_rep},
Eq.~(F5).
The corresponding completely positive trace-non-increasing map is
$\rho'=\sum_\sigma K_\sigma\rho K_\sigma^\dagger$.

As shown in Appendix~F, this channel removes coherences between
different permutation sectors while preserving sector-resolved
probabilities, yielding
\begin{equation}
  \rho'
  \;=\;
  \sum_\sigma
  \frac{|A_\sigma|^2}
       {|M|^2\,\displaystyle\prod_{r=1}^{n-1}|\alpha_{\sigma(r)}|^2}
  \;|\sigma\rangle\langle\sigma|\otimes|k_\sigma\rangle\langle k_\sigma|.
  \label{eq:rho_prime}
\end{equation}
Using the proportionality relation~\eqref{eq:psi_Asigma}, the
diagonal entries satisfy
\begin{equation}
  (\rho')_{\sigma\sigma} \;\propto\; |A_\sigma|^2,
  \label{eq:rho_diag}
\end{equation}
establishing a direct correspondence between the quantum-walk
probabilities and the squared Parke--Taylor amplitudes.
This channel description complements the unitary walk evolution
and establishes a connection with open quantum
walks~\cite{Attal2012,Sinayskiy2012}.

\section{Coherent Amplitude Reconstruction} \label{sec:CAR}

\noindent The Kraus-channel construction extracts the individual Parke--Taylor contributions associated with each terminal permutation sector. To recover the full scattering amplitude, these sectors must be recombined coherently.

\subsection{Physical origin of the weighted collection operator}

The need for a separate collection step arises from a structural asymmetry between the kinematic and color contributions to the scattering amplitude. The kinematic structure --- the ordered cyclic product of spinor brackets forming the Parke--Taylor denominator --- is \emph{local}: each factor $\langle\ell(v)\;j'\rangle^{-1}$ depends only on the current node and the particle being appended. It is therefore accumulated naturally, edge by edge, during the quantum walk through the coin operator weights $v_{v,u_j}$. This is why the walk state $|\psi_{\mathrm{fin}}\rangle$ already encodes the kinematic partial amplitudes $A_\sigma$ (up to the normalization factors $\alpha_{\sigma(r)}$), as established in~\eqref{eq:psi_Asigma}.

The color structure, by contrast, is \emph{global}: the trace factor $\mathrm{Tr}(T^{a_1}T^{a_{\sigma(2)}}\cdots T^{a_{\sigma(n)}})$ depends on the \emph{complete} permutation $\sigma$, not on any single edge. Consequently, color weights cannot be assigned locally during the walk --- they require knowledge of the full root-to-terminal path, which is available only after the walker has reached a terminal node.

The weighted collection operator $W^T$ resolves this asymmetry. It acts exclusively on the terminal sectors and performs two operations simultaneously: (i) it attaches the color weight $\mathrm{Tr}(T^{a_1}T^{a_{\sigma(2)}}\cdots T^{a_{\sigma(n)}})$ and any remaining kinematic correction to each permutation sector, and (ii) it transfers all terminal sectors to a common reference node $R$, enabling subsequent interference. The transfer to a common node is necessary because quantum-mechanical interference requires superposed amplitudes to occupy the same spatial mode; without collection, the amplitudes reside at distinct terminal nodes and cannot interfere.

After collection, all path-dependent information is stored in the coin degree of freedom. The coin state $|k_\sigma\rangle$ uniquely identifies which permutation sector contributed each amplitude, so a quantum Fourier transform on the coin space produces the coherent superposition required for the full color-decomposed amplitude.

\subsection{Action on the terminal walk state}

The explicit construction of $W^T$ is given in Appendix~\ref{WCO}.
Acting on the final walk state,
\begin{equation}
W^T \ket{\psi_{\mathrm{fin}}} =
|R\rangle
\otimes
\sum_{\sigma\in\mathcal T}
c_\sigma
\,
|k_\sigma\rangle,
\label{eq:collection_main}
\end{equation} where
\begin{equation}
c_\sigma
=
\frac{
\mathrm{Tr}
\!\left(
T^{a_1}
T^{a_{\sigma(2)}}
\cdots
T^{a_{\sigma(n)}}
\right)
}
{\mathcal N}
\,
\frac{
A_\sigma
}
{
\prod_r
\alpha_{\sigma(r)}
}.
\end{equation}

\begin{remark}[Normalization factors]
\label{rem:normalization}
The coefficients $c_\sigma$ contain the factor $1/\prod_r\alpha_{\sigma(r)}$, which varies across permutation sectors because the local normalization constants $\alpha_v$ depend on the spinor configuration at each node. These factors arise from the requirement that each local coin operator $C_v$ be unitary: the physical transition amplitudes $w(v,u_j)$ are in general non-normalized, and the rescaling by $\alpha_v$ is the cost of embedding them into a unitary matrix. While these factors do not cancel in the coherent sum, they are fully determined by the input kinematics and can be computed classically in $O(n!)$ operations or absorbed into the normalization constant $\mathcal{N}$.
\end{remark}

\begin{remark}[Implementability of $W^T$]
\label{rem:implementability}
The weighted collection operator $W^T$ is bounded ($|a_\sigma|\le 1$) but not unitary, since it maps $|\mathcal{T}|=(n-1)!$ orthogonal terminal states to a single reference node. In a physical quantum implementation, $W^T$ would need to be realized through an ancilla-assisted unitary dilation~\cite{Nielsen2010}, embedding the non-unitary map into a larger unitary operator on an extended Hilbert space, followed by post-selection on the ancilla. The success probability of this post-selection decreases with the number of permutation sectors, which is a limitation shared by other amplitude-embedding approaches including the quantum circuit construction of Ref.~\cite{Bashore2025}.
\end{remark}

Thus, all terminal permutation sectors are collected at the same reference node while the interference information remains encoded in the coin degrees of freedom. To recover the coherent sum over permutation sectors we apply the Quantum Fourier Transform (QFT)~\cite{Shor1994,Nielsen2010} on the coin space and project onto the Fourier vacuum mode $|0\rangle$.

\begin{remark}[Choice of mixing unitary]
The reconstruction requires projecting onto a coin state that has uniform overlap with all coin basis states. Since $\langle0|U_{\mathrm{QFT}}|k\rangle=1/\sqrt{d}$ for every $k$, the QFT followed by projection onto $|0\rangle$ produces an equal-weight coherent sum over all permutation sectors. Any $d\times d$ unitary whose first row is $(1/\sqrt{d},\ldots,1/\sqrt{d})$ would achieve the same result; the QFT is adopted as the standard canonical choice.
\end{remark}

The resulting coherent color-weighted sum satisfies
\begin{equation}
|\mathcal M_{\rm QW}|^2
\propto
|\sum_{\sigma\in S_{n-1}}
\mathrm{Tr}
\left(
T^{a_1}
T^{a_{\sigma(2)}}
\cdots
T^{a_{\sigma(n)}}
\right)
A_\sigma|^2.
\end{equation} The detailed derivation of the weighted collection operator, its action on the terminal walk state, and the Quantum Fourier Transform reconstruction are presented in Appendices~\ref{WCO} and~\ref{QFT_AR}.

\section{Numerical Validation} \label{sec:NV}

In this section, we present numerical results for the four-gluon MHV scattering amplitude obtained using the quantum walk framework developed in the previous sections. The simulations serve three complementary purposes. First, they demonstrate the propagation of probability amplitude through the permutation tree. Second, they illustrate the extraction of sector-resolved Parke--Taylor weights through the Kraus-channel formulation. Finally, they validate the numerical reconstruction of the color-ordered scattering contributions.

\subsection{Quantum Walk on the Permutation Tree}

The walk is initialized at the root of the permutation tree according to

\begin{equation}
\ket{\psi_0}=\ket{(1)}\otimes
\ket{0},
\label{eq:initial_state_num}
\end{equation}

where the first register labels the position of the walker on the permutation graph and the second register denotes the coin degree of freedom.

The evolution is generated through repeated applications of the unitary operator

\begin{equation}
U = SC,
\label{eq:walk_operator_num}
\end{equation}

where $C$ denotes the local coin operation and $S$ is the conditional shift operator. After $t$ steps the state of the walker is

\begin{equation}
\ket{\psi_t}=U^t\ket{\psi_0}.\label{eq:walk_state_t}
\end{equation}

For the four-gluon case, the walker reaches the terminal sectors after three steps. Figure~\ref{fig:permutation_tree} shows the permutation tree employed in the simulations. Each root-to-terminal path corresponds to a unique color ordering.


The probability distribution at time $t$ is obtained from

\begin{equation}
P_t(i)=\sum_{c}\left|\langle i,c|\psi_t\rangle
\right|^2,
\label{eq:position_probability}
\end{equation}

where $i$ labels the position basis and $c$ labels the coin state. The evolution of the walker probability distribution is shown in Fig.~\ref{fig:walk_evolution}.

\begin{figure}[h]
\centering
\includegraphics[
width=0.5\textwidth,
keepaspectratio
]{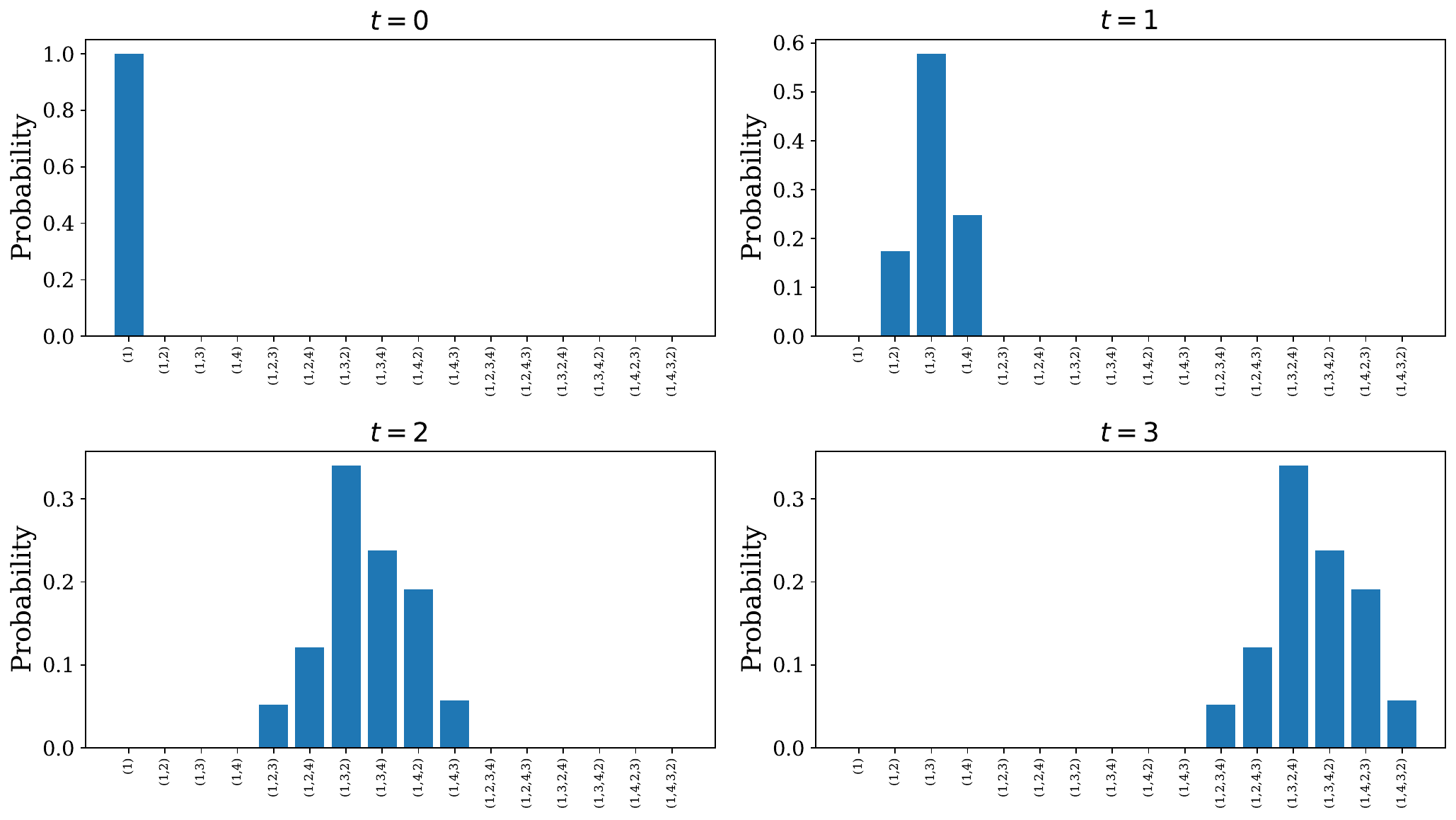}
\caption{
Probability distribution of the quantum walker after successive applications of the evolution operator $U=SC$. The final step corresponds to the terminal permutation sectors.
}
\label{fig:walk_evolution}
\end{figure}

\subsection{Terminal Permutation Sectors}

After three walk steps the state takes the form

\begin{equation}
\ket{\psi_{\mathrm{fin}}}=\sum_{\sigma\in\mathcal T}\psi_\sigma
\ket{\sigma}
\otimes
\ket{k_\sigma},
\label{eq:final_state_num}
\end{equation}

where $\mathcal T$ denotes the set of terminal permutation sectors and $\psi_\sigma$ are the corresponding complex amplitudes.

The occupation probability associated with a terminal sector $\sigma$ is

\begin{equation}
P(\sigma)=|\psi_\sigma|^2.
\label{eq:terminal_probability}
\end{equation}

Figure~\ref{fig:terminal_distribution} shows the resulting terminal-state distribution. Only terminal sectors possess nonzero probability, confirming that the walk correctly generates the complete set of color orderings.

\begin{figure}[h]
\centering
\includegraphics[
width=0.5\textwidth,
keepaspectratio
]{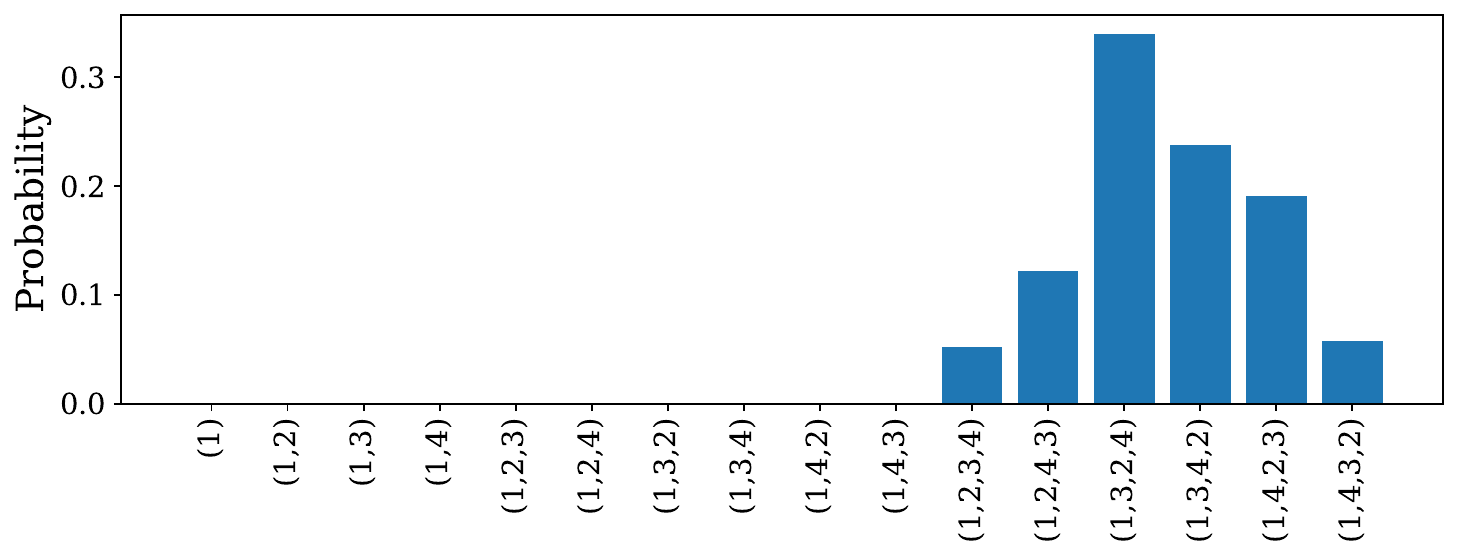}
\caption{
Occupation probabilities of the terminal permutation sectors after completion of the quantum walk.
}
\label{fig:terminal_distribution}
\end{figure}

The amplitudes $\psi_\sigma$ retain the relative phases accumulated during the walk and therefore encode the interference structure required for scattering-amplitude reconstruction.

\subsection{Extraction of Parke--Taylor Weights}

To extract sector-resolved probabilities we construct the density matrix

\begin{equation}
\rho=\ket{\psi_{\mathrm{fin}}}\bra{\psi_{\mathrm{fin}}}.
\label{eq:density_before_num}
\end{equation}

Figure~\ref{fig:density_before} shows the density matrix restricted to the occupied terminal sectors. The visible off-diagonal elements indicate coherence between different permutation sectors.

\begin{figure}[h]
\centering
\includegraphics[
width=0.5\textwidth,
keepaspectratio
]{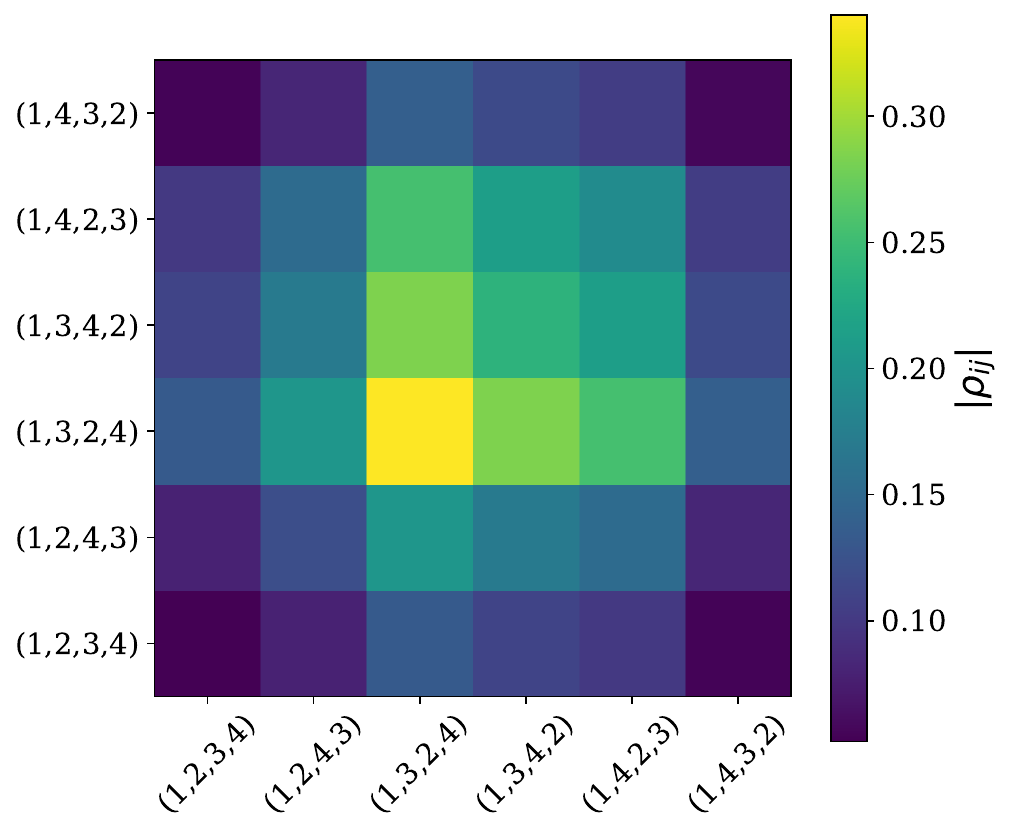}
\caption{
Density matrix of the coherent terminal-state superposition before application of the Kraus channel.
}
\label{fig:density_before}
\end{figure}

The extraction channel is defined through

\begin{equation}
\rho'=\sum_{\sigma\in\mathcal T}K_\sigma\rho K_\sigma^\dagger,
\label{eq:kraus_channel_num}
\end{equation}

where $K_\sigma$ are the Kraus operators introduced in Sec.\ref{pkt}.

Application of the channel removes the inter-sector coherences while preserving the sector weights. The resulting density matrix is shown in Fig.~\ref{fig:density_after}.

\begin{figure}[h]
\centering
\includegraphics[
width=0.5\textwidth,
keepaspectratio
]{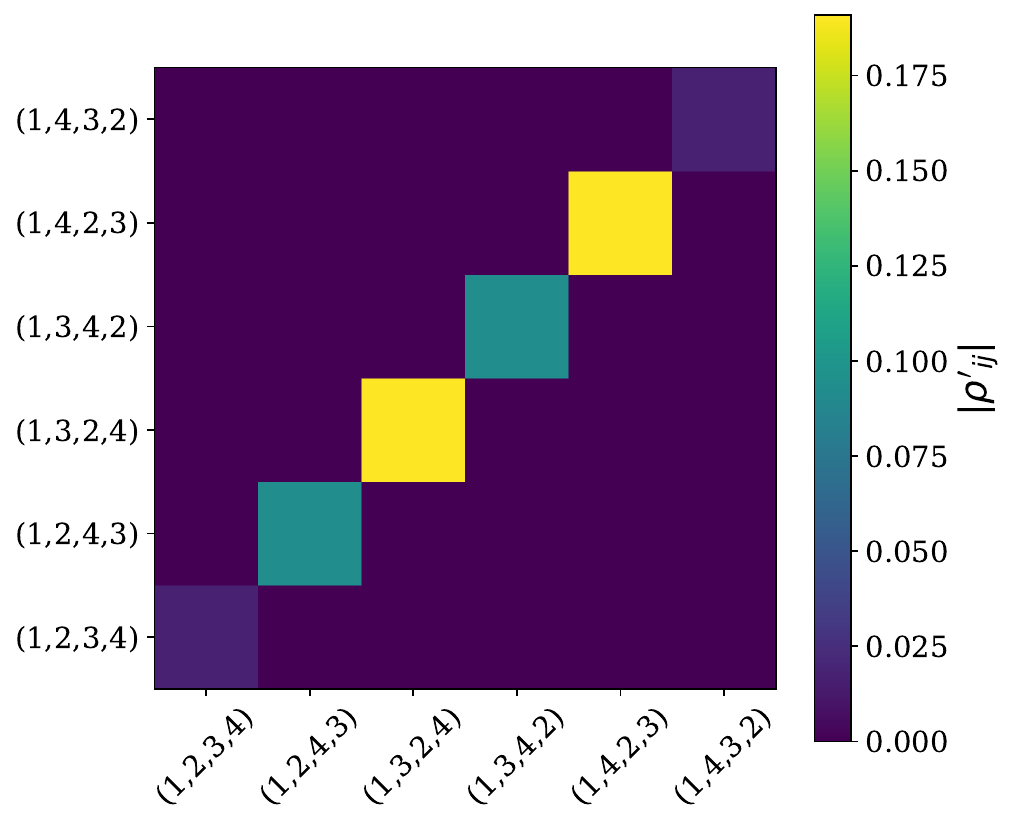}
\caption{
Density matrix after application of the Kraus channel. The off-diagonal coherences are suppressed and only the sector probabilities remain.
}
\label{fig:density_after}
\end{figure}

The diagonal elements of the transformed density matrix satisfy

\begin{equation}
(\rho')_{\sigma\sigma}=\frac{
|A_\sigma|^2
}{
|M|^2
\prod_{r=1}^{n-1}
|\alpha_{\sigma(r)}|^2
},
\label{eq:pt_weight_relation}
\end{equation}

where $A_\sigma$ denotes the Parke--Taylor contribution associated with the permutation sector $\sigma$.

The extracted sector probabilities therefore provide direct access to the relative Parke--Taylor weights encoded within the quantum walk.

\subsection{Numerical simulation of quantum walk and comparison with analytical amplitudes}

The extracted sector weights are compared with the exact Parke--Taylor values in Fig.~\ref{fig:amplitude_comparison}. Excellent agreement is observed across all six permutation sectors.

\begin{figure}[h]
\centering
\includegraphics[
width=0.5\textwidth,
keepaspectratio
]{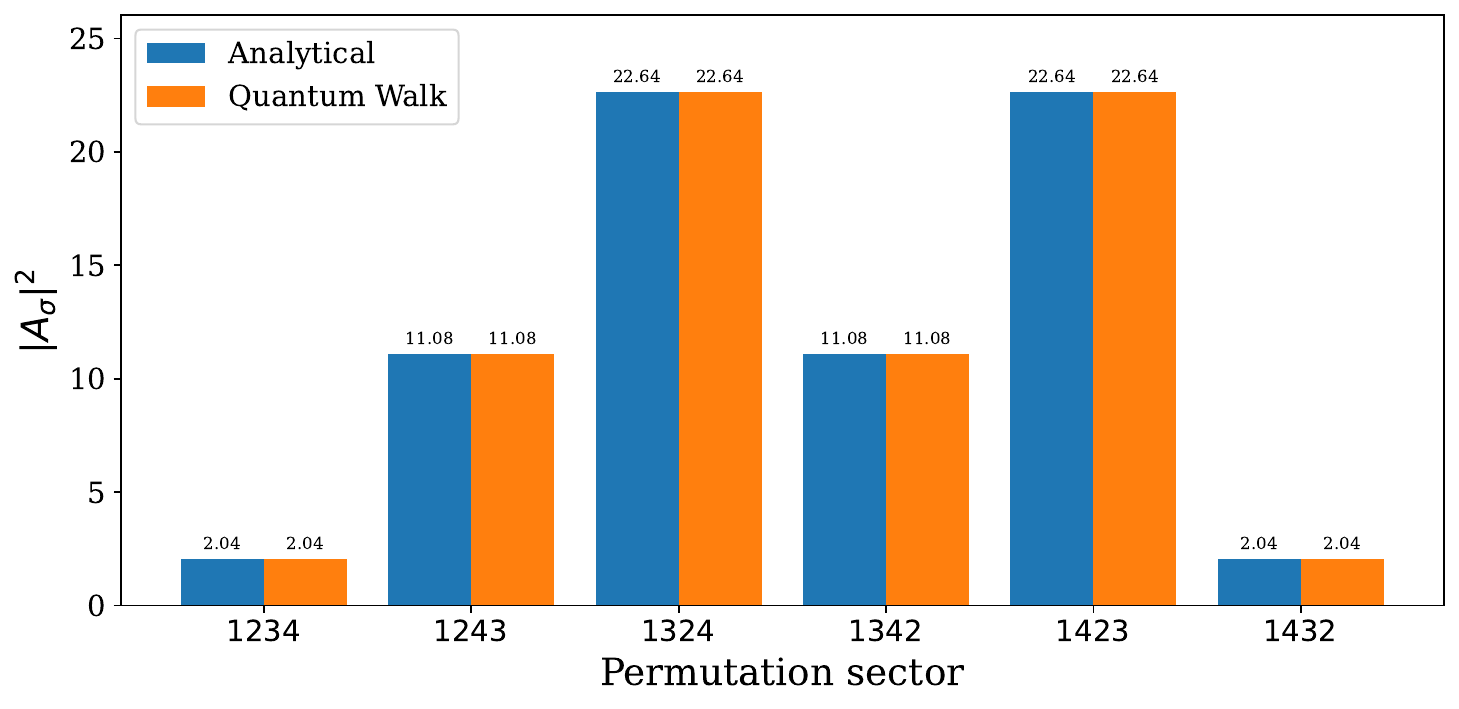}
\caption{
Comparison between analytical Parke--Taylor weights and the values reconstructed from the quantum walk framework for all terminal permutation sectors.
}
\label{fig:amplitude_comparison}
\end{figure}

The reconstructed values reproduce the analytical amplitudes with sub-percent relative error for every sector, demonstrating that the quantum walk correctly encodes both the combinatorial structure of the permutation tree and the associated kinematic weights.

\subsection{Coherent Reconstruction of the Scattering Amplitude}

The Kraus-channel formulation discussed above extracts the sector-resolved probabilities associated with the individual Parke--Taylor contributions. To recover the full scattering amplitude, the terminal permutation sectors must be recombined coherently.

Starting from the final walk state

\begin{equation}
\ket{\psi_{\mathrm{fin}}}
=
\sum_{\sigma\in\mathcal T}
\psi_\sigma
\,
\ket{\sigma}
\otimes
\ket{k_\sigma},
\end{equation}

we apply the weighted collection operator \(W^T\), which maps all terminal sectors to a common reference node \(R\) while preserving their relative amplitudes.

\begin{figure}[h]
\centering
\includegraphics[
width=0.4\textwidth,
keepaspectratio
]{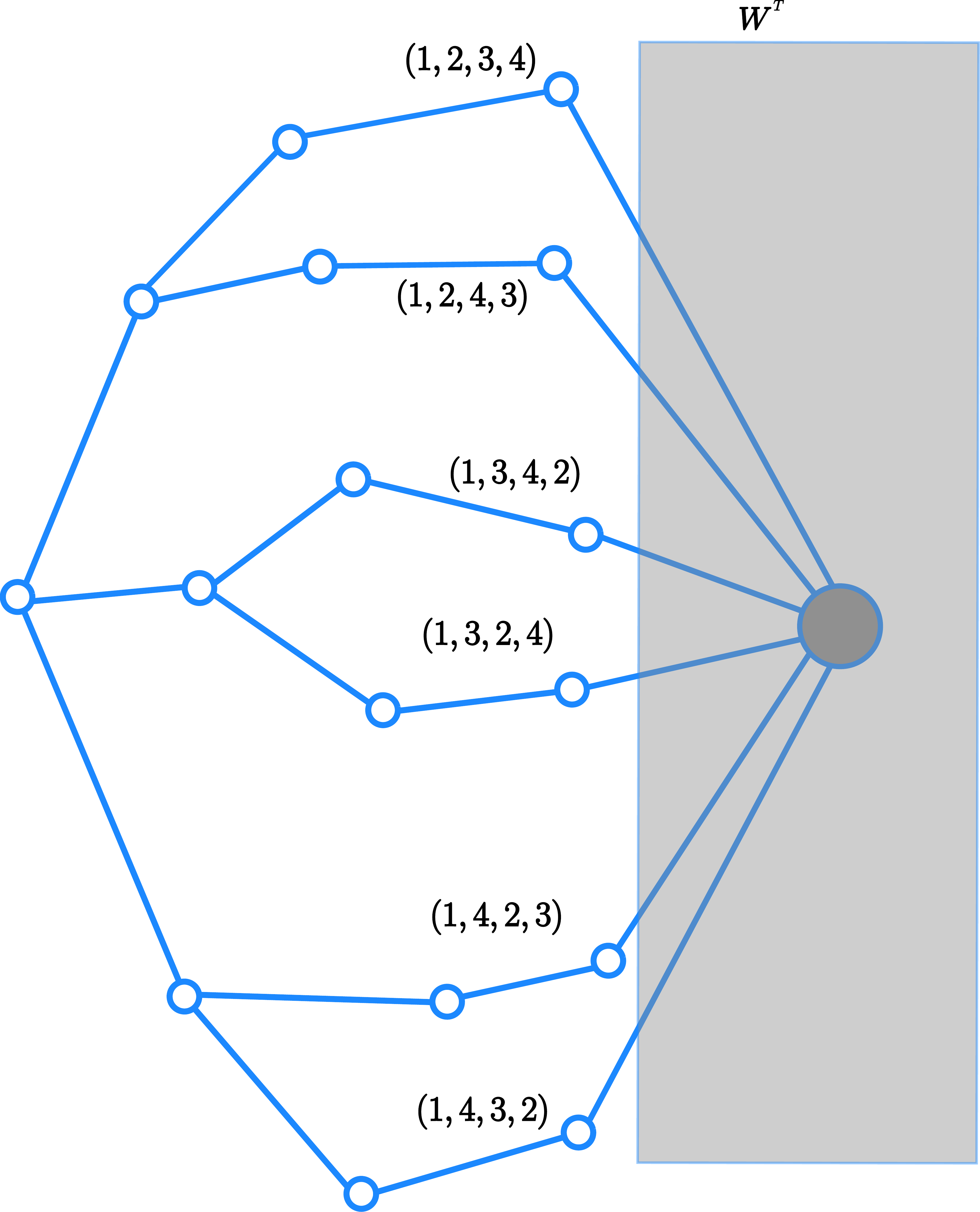}
\caption{
Graphical representation of the weighted collection procedure.
The terminal permutation sectors of the four-gluon permutation tree are
coherently mapped to a common reference node \(R\). The weights attached
to each branch preserve the Parke--Taylor information and enable their
subsequent interference.
}
\label{fig:collection_graph}
\end{figure}

The resulting state takes the form

\begin{equation}
W^T
\ket{\psi_{\mathrm{fin}}}
=
|R\rangle
\otimes
\sum_{\sigma\in\mathcal T}
c_\sigma
\,
|k_\sigma\rangle ,
\label{eq:collected_state_num}
\end{equation}

where the coefficients \(c_\sigma\) contain the color and kinematic information associated with the corresponding permutation sector. As illustrated in Fig.~\ref{fig:collection_graph}, the collection operator maps all terminal sectors to a single reference node while preserving their
relative weights.

To generate interference between the collected contributions, a quantum Fourier transform~\cite{Shor1994,Nielsen2010} is applied to the coin space. For a coin dimension \(d\), the Fourier operator is

\begin{equation}
U_{\mathrm{QFT}}
=
\frac{1}{\sqrt d}
\sum_{j,k=0}^{d-1}
e^{2\pi i jk/d}
|j\rangle\langle k|.
\label{eq:qft_num}
\end{equation}

For the four-gluon example considered here, \(d=3\). The Fourier transform mixes the amplitudes stored in the coin basis and produces the coherent superposition required for amplitude reconstruction.

The reconstructed scattering amplitude is obtained by projecting onto the Fourier vacuum mode,

\begin{equation}
\mathcal M_{\mathrm{QW}}
=
\sqrt d\,
\langle 0 |
U_{\mathrm{QFT}}
W^T
\ket{\psi_{\mathrm{fin}}},
\label{eq:reconstructed_amplitude}
\end{equation}
where the factor \(\sqrt d\) compensates for the normalization of the Fourier transform.
\begin{remark}[Phase information]
\label{rem:phase}
The projection formula~\eqref{eq:reconstructed_amplitude} reconstructs the
\emph{complex} amplitude $\mathcal{M}_{\mathrm{QW}}\in\mathbb{C}$.
The relative phases between permutation sectors are preserved
throughout the walk in the coin-space superposition, and are
recovered by the QFT projection.  The squared modulus
$|\mathcal{M}_{\mathrm{QW}}|^2$ is reported in Table~\ref{tab:amplitude_reconstruction} for
comparison with analytical results, but the construction yields
the full complex amplitude.
\end{remark}
The resulting quantity corresponds to the coherent sum of the color-ordered contributions encoded in the quantum walk. Numerical evaluation confirms that the reconstructed amplitude agrees with the analytical Parke--Taylor result within numerical precision.

\begin{table}[h]
\centering
\begin{tabular}{lcc}
\hline
Quantity & Analytical & Quantum Walk \\
\hline

$|\mathcal M|^2$
&
$0.05634$
&
$0.05634$
\\

Relative Error (\%)
&
\multicolumn{2}{c}{0.00}
\\
\hline
\end{tabular}
\caption{
Comparison between the analytical scattering amplitude and the value reconstructed from the quantum walk framework.
}
\label{tab:amplitude_reconstruction}
\end{table}

\section{Conclusion}
We have introduced a quantum-walk representation of color-ordered gluon scattering amplitudes based on the permutation-tree structure underlying the Parke–Taylor formula. The construction maps the recursive generation of color orderings onto a discrete-time quantum walk defined on a permutation graph: coin operators encode local branching probabilities, while shift operators propagate amplitude through the tree, so that the resulting coherent walk state contains contributions from all color-ordered sectors simultaneously. A quantum-channel formulation was developed to extract the sector-resolved probabilities associated with individual Parke–Taylor contributions, and numerical simulations for four-gluon MHV amplitudes show excellent agreement with analytical results, confirming that the quantum walk faithfully reproduces the expected kinematic weights.

These results establish a connection between scattering amplitudes, quantum walks, and open quantum-system techniques and quantum network architecture.  The quantum algorithmic framework of studying scattering amplitudes  provides a quantum-information perspective to model, simulate and study such dynamics using quantum computers. Future directions include extensions to higher multiplicities, non-MHV amplitudes, loop corrections, and implementations using open quantum walks and quantum-network architectures.

\newpage

\appendix

\section*{Supplementary Material}

This supplementary material provides additional implementation details and numerical data supporting the results presented in the main text. In particular, we present the bijectivity proof for the shift operator, summarize the basis-indexing conventions, construction of the local coin and shift operators, implementation of the weighted collection operator, verification of the Kraus channel, and additional numerical results beyond those reported in the main manuscript.

\section{Bijectivity of the Shift Operator}
\label{sec:counting_proof}
\label{sec:unitarity_proof}

We provide additional detail for the bijectivity argument
(Proposition~\ref{prop:S_unitary}).

Under the global edge-labeling convention, coin value $k$
corresponds to particle $p=k+2$.
For each sector $k$, the basis states $V\times\{k\}$ are
partitioned into $F_k$ (forward), $L_k$ (self-loop), and
$D_k$ (displaced); see~\eqref{eq:Fk}--\eqref{eq:Dk}.

\paragraph{Source and target classification.}
The three rules produce three disjoint sets of images:

\begin{enumerate}
\item \textit{Forward images:}
$\mathrm{Im}(F_k)=\{(u,k):u\neq r,\;\kappa(u)=k\}$,
i.e.\ all non-root nodes whose last element is particle $k{+}2$.
The forward rule is injective (each such node has a unique parent).

\item \textit{Self-loop images:}
$\mathrm{Im}(L_k)=L_k$ (each state maps to itself).
These satisfy $k\neq\kappa(v)$, so they are disjoint from
forward images (which satisfy $\kappa(u)=k$).

\item \textit{Displaced images:}
$\mathrm{Im}(D_k)=\mathcal{U}_k$, the uncovered targets.
By definition, $\mathcal{U}_k$ contains neither forward images
nor self-loop states, so this set is disjoint from the other two.
\end{enumerate}

Since the three image sets are disjoint and each rule is
injective on its domain, $\tau$ is injective on $V\times\{k\}$.
By the pigeonhole principle on a finite set, injectivity
implies bijectivity.
Since this holds for every $k$, $\tau$ is a permutation on the
full basis and $S$ is unitary.

\section{Basis Indexing Convention}

The quantum walk is implemented on the Hilbert space

\begin{equation}
\mathcal H
=
\mathcal H_{\mathrm{pos}}
\otimes
\mathcal H_{\mathrm{coin}},
\end{equation}

with

\begin{equation}
\dim(\mathcal H_{\mathrm{pos}})=16,
\qquad
\dim(\mathcal H_{\mathrm{coin}})=3.
\end{equation}

The computational basis is ordered according to

\begin{equation}
|0,0\rangle,
|0,1\rangle,
|0,2\rangle,
|1,0\rangle,
|1,1\rangle,
|1,2\rangle,
\ldots,
|15,2\rangle,
\end{equation}

resulting in a total Hilbert-space dimension of \(48\).

Table~\ref{tab:basis_map} summarizes the mapping between node labels and permutation sectors.

\begin{table}[h]
\centering
\begin{tabular}{cc}
\hline
Node Index & Permutation Sector \\
\hline
0 & $(1)$ \\
1 & $(12)$ \\
2 & $(13)$ \\
3 & $(14)$ \\
4 & $(123)$ \\
5 & $(124)$ \\
6 & $(132)$ \\
7 & $(134)$ \\
8 & $(142)$ \\
9 & $(143)$ \\
10 & $(1234)$ \\
11 & $(1243)$ \\
12 & $(1324)$ \\
13 & $(1342)$ \\
14 & $(1423)$ \\
15 & $(1432)$ \\
\hline
\end{tabular}
\caption{Basis indexing convention used in the numerical implementation.}
\label{tab:basis_map}
\end{table}

\section{Classification of Basis States in the Shift Operator}
\label{app:shift_classification}

For the four-gluon permutation tree, every basis state
$|v,k\rangle\in V\times\{0,1,2\}$ belongs uniquely to one of the
three subsets $F_k$, $L_k$, or $D_k$ introduced in
Sec.~\ref{ssec:shift}. The classification is determined entirely by
the remaining-particle set $R(v)$, the incoming coin label
$\kappa(v)$, and whether the node is terminal.
Table~\ref{tab:shift_classification} lists the resulting partition
explicitly for all sixteen nodes of the permutation tree.

\begin{table}[t]
\centering
\caption{Classification of the basis states for the four-gluon
permutation tree. Here $F$, $L$, and $D$ denote the forward,
loop and displaced subsets, respectively.}
\label{tab:shift_classification}

\renewcommand{\arraystretch}{1.15}

\begin{tabular}{c|ccc}
\hline
Node $v$
&
$k=0$
&
$k=1$
&
$k=2$
\\
\hline

$(1)$
&
$F$
&
$F$
&
$F$
\\

$(1,2)$
&
$D$
&
$F$
&
$F$
\\

$(1,3)$
&
$F$
&
$D$
&
$F$
\\

$(1,4)$
&
$F$
&
$F$
&
$D$
\\

$(1,2,3)$
&
$L$
&
$D$
&
$F$
\\

$(1,2,4)$
&
$L$
&
$F$
&
$D$
\\

$(1,3,2)$
&
$D$
&
$L$
&
$F$
\\

$(1,3,4)$
&
$F$
&
$L$
&
$D$
\\

$(1,4,2)$
&
$D$
&
$F$
&
$L$
\\

$(1,4,3)$
&
$F$
&
$D$
&
$L$
\\

$(1,2,3,4)$
&
$D$
&
$D$
&
$D$
\\

$(1,2,4,3)$
&
$D$
&
$D$
&
$D$
\\

$(1,3,2,4)$
&
$D$
&
$D$
&
$D$
\\

$(1,3,4,2)$
&
$D$
&
$D$
&
$D$
\\

$(1,4,2,3)$
&
$D$
&
$D$
&
$D$
\\

$(1,4,3,2)$
&
$D$
&
$D$
&
$D$
\\

\hline
\end{tabular}

\end{table}

The cardinalities of the three subsets for each coin sector are
summarized in Table~\ref{tab:shift_counts}. As expected,

\begin{equation}
|F_k|+|L_k|+|D_k|
=
|V|,
\end{equation}

for every coin value $k$, confirming that the three subsets form a
partition of the node set in each coin sector.
\begin{table}[h]
\centering
\caption{Cardinalities of the three subsets for the four-gluon
example.}
\label{tab:shift_counts}

\renewcommand{\arraystretch}{1.15}

\begin{tabular}{c|ccc}
\hline
Subset
&
$k=0$
&
$k=1$
&
$k=2$
\\
\hline

$F_k$
&
6
&
6
&
6
\\

$L_k$
&
3
&
3
&
3
\\

$D_k$
&
7
&
7
&
7
\\

\hline

Total
&
16
&
16
&
16
\\

\hline
\end{tabular}

\end{table}

\section{Quantum Fourier Transform}

For the three-dimensional coin space used in the four-gluon example, the quantum Fourier transform is

\begin{equation}
U_{\mathrm{QFT}}
=
\frac{1}{\sqrt3}
\begin{pmatrix}
1 & 1 & 1 \\
1 & \omega & \omega^2 \\
1 & \omega^2 & \omega
\end{pmatrix},
\end{equation}

where

\begin{equation}
\omega=e^{2\pi i/3}.
\end{equation}

After application of the collection operator, the Fourier transform generates the interference required for reconstruction of the color-ordered scattering amplitude.

\section{Five-Gluon Permutation Graph}

The main text focuses on the four-gluon scattering process. To illustrate the scalability of the framework, Fig.~\ref{fig:five_gluon_tree} shows the corresponding graph for the five-gluon case.

\begin{figure}[h]
\centering
\includegraphics[
width=0.5\textwidth,
keepaspectratio
]{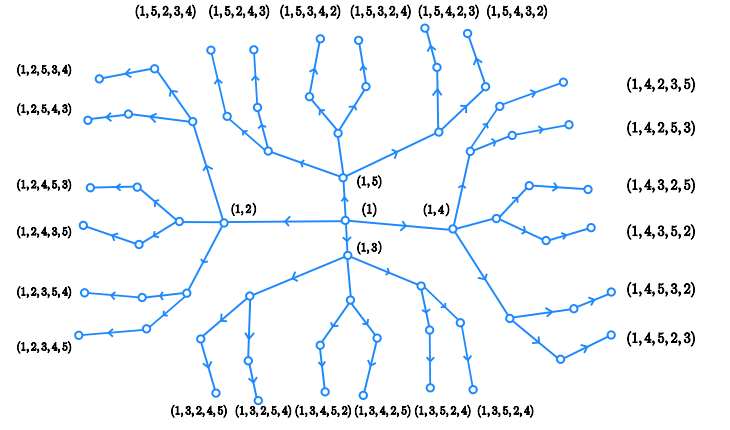}
\caption{
Permutation graph for the five-gluon scattering process.
The graph is generated using the same recursive insertion procedure used for the four-gluon example and illustrates the growth of the permutation-tree structure with increasing multiplicity.
}
\label{fig:five_gluon_tree}
\end{figure}

\section{Growth of the Permutation Tree}

For an $n$-gluon scattering process, the permutation tree is generated by
successively inserting a new particle into all admissible positions of the
existing orderings. The number of terminal permutation sectors is therefore

\begin{equation}
N_{\mathrm{term}}=(n-1)!.
\end{equation}

The total number of nodes in the permutation tree is

\begin{equation}
g(n)
=
\sum_{k=0}^{n-1}
\frac{(n-1)!}{k!},
\label{eq:total_nodes}
\end{equation}

while the Hilbert-space dimension of the coined quantum walk is

\begin{equation}
D(n)
=
(n-1)\,g(n),
\label{eq:hilbert_dimension}
\end{equation}

where the factor $(n-1)$ arises from the coin dimension.

Table~\ref{tab:growth} summarizes the scaling of the permutation-tree
representation and the corresponding quantum-walk Hilbert space.

\begin{table}[h]
\centering
\small
\begin{tabular}{ccc}
\hline
$n$ & $g(n)$ & $D(n)$ \\
\hline
4 & 16 & 48 \\
5 & 65 & 260 \\
6 & 326 & 1630 \\
\hline
\end{tabular}
\caption{
Growth of the permutation-tree graph and the corresponding quantum-walk
Hilbert-space dimension.
}
\label{tab:growth}
\end{table}

\section{Derivation of the Quantum-Walk State} \label{coh_evo}

We derive the explicit form of the quantum-walk state generated
by $(SC)^{n-1}|\psi_0\rangle$ and establish its correspondence
to the Parke--Taylor denominator.  The derivation is fully
consistent with the coin condition~\eqref{eq:coin_column} of the
main text: the physically relevant column of $C_v$ is always
column $\kappa(v)$, equal to the coin value carried by the walker
upon arriving at $v$.

\paragraph{Step 0 (Initialization).}
\begin{equation}
  |\psi_0\rangle = |(1)\rangle\otimes|0\rangle.
  \tag{E1}
\end{equation}
The root has $\kappa(r)=0$ by convention, so the coin operator
$C_{(1)}$ is applied through its column~$0$.

\paragraph{Step 1.}
\begin{align}
  C|\psi_0\rangle
  &= \sum_{k_1=0}^{d-1}
     \bigl(C_{(1)}\bigr)_{k_1,\,0}
     \;|(1)\rangle\otimes|k_1\rangle,
  \tag{E2}\\[4pt]
  SC|\psi_0\rangle
  &= \sum_{k_1=0}^{d-1}
     \bigl(C_{(1)}\bigr)_{k_1,\,0}
     \;|(1,k_1{+}2)\rangle\otimes|k_1\rangle,
  \tag{E3}
\end{align}
where the shift appends particle $k_1{+}2$ (the particle
corresponding to coin $k_1$ under the global map).
The child $(1,k_1{+}2)$ has $\kappa\bigl((1,k_1{+}2)\bigr)
=c(k_1{+}2)=k_1$.

\paragraph{General step $m$.}
After $m$ applications, with $v_0=r$ and
$v_s=(1,k_1{+}2,\ldots,k_s{+}2)$ for $s\ge 1$:
\begin{equation}
  (SC)^m|\psi_0\rangle
  = \sum_{k_1,\ldots,k_m}
    \left(\prod_{s=1}^{m}
      \bigl(C_{v_{s-1}}\bigr)_{k_s,\;\kappa(v_{s-1})}
    \right)
    |v_m\rangle\otimes|k_m\rangle.
  \tag{E4}
\end{equation}
The column index $\kappa(v_{s-1})=k_{s-1}$ (where $k_0=0$) is
the coin label with which the walker arrived at $v_{s-1}$.
Under the global convention, $k_{s-1}=c(\sigma(s))=\sigma(s)-2$
for $s\ge 2$ and $k_0=0$ for the root.
This is guaranteed by the shift operator, which preserves the
coin label: the walker arrives at $v_{s-1}$ in state
$|k_{s-1}\rangle$, the coin operator maps $|k_{s-1}\rangle$
through column $\kappa(v_{s-1})=k_{s-1}$ of $C_{v_{s-1}}$,
and the shift routes to the child indexed by the new coin value
$k_s=c(\sigma(s{+}1))=\sigma(s{+}1)-2$.

\paragraph{Terminal state ($m=n-1$).}
Each sequence $(k_1,\ldots,k_{n-1})$ with
$k_s\in\{0,\ldots,|R(v_{s-1})|-1\}$ specifies a unique
root-to-terminal path and hence a unique permutation
$\sigma=(1,\sigma(2),\ldots,\sigma(n))\in S_{n-1}$.
The map from valid coin sequences to permutations is a bijection.
Re-indexing:
\begin{equation}
  |\psi_{\mathrm{fin}}\rangle
  = \sum_{\sigma\in S_{n-1}}
    \psi_\sigma\,|\sigma\rangle\otimes|\kappa_\sigma\rangle,
  \tag{E5}
\end{equation}
where $\kappa_\sigma=k_{n-1}=c(\sigma(n))=\sigma(n)-2$ is the
global coin label of the last appended particle.

\paragraph{Explicit path amplitude.}
From condition~\eqref{eq:coin_column} and the global coin
map~\eqref{eq:global_coin_map}, at step $s$ the relevant matrix
element is
\begin{equation}
  \bigl(C_{v_{s-1}}\bigr)_{c(\sigma(s{+}1)),\,\kappa(v_{s-1})}
  = \tilde{v}_{v_{s-1},\,\sigma(s{+}1)}
  = \frac{1}{\alpha_{v_{s-1}}}
    \cdot
    \frac{1}{\langle\sigma(s)\;\sigma(s{+}1)\rangle},
  \tag{E6}
\end{equation}
where the row index $c(\sigma(s{+}1))=\sigma(s{+}1)-2$ is the
global coin label of the appended particle, and
$\ell(v_{s-1})=\sigma(s)$.
Taking the product over $s=1,\ldots,n-1$:
\begin{equation}
  \psi_\sigma
  = \left(\prod_{r=1}^{n-1}\frac{1}{\alpha_{\sigma(r)}}\right)
    \cdot
    \frac{1}{\langle 1\,\sigma(2)\rangle
              \langle\sigma(2)\,\sigma(3)\rangle
              \cdots
              \langle\sigma(n)\,1\rangle}.
  \tag{E7}
\end{equation}
Comparing with
\begin{equation}
  A_\sigma
  = \frac{\langle 1\tilde{k}\rangle^4}
         {\langle 1\,\sigma(2)\rangle
          \langle\sigma(2)\,\sigma(3)\rangle
          \cdots
          \langle\sigma(n)\,1\rangle},
  \tag{E8}
\end{equation}
yields
\begin{equation}
  \boxed{\;
  \psi_\sigma
  = \frac{\langle\sigma(n)\,1\rangle}
         {\langle 1\tilde{k}\rangle^4\;
          \displaystyle\prod_{r=1}^{n-1}\alpha_{\sigma(r)}}
  \;A_\sigma.
  \;}
  \tag{E9}
\end{equation}
The coined quantum walk thus reproduces the ordered
Parke--Taylor denominator through coherent evolution on the
permutation tree, with the proportionality constant explicitly
known.

\section{Derivation of the Kraus-Channel Representation} \label{KC_rep}

The final coherent walk state derived in Appendix \ref{coh_evo} is

\begin{equation}
\ket{\psi_{\mathrm{fin}}}
=
\sum_{\sigma\in\mathcal T}
\psi_\sigma
\,
\ket{\sigma}
\otimes
\ket{k_\sigma},
\end{equation}

where

\begin{equation}
\psi_\sigma
=
A_\sigma
\frac{
\langle\sigma(n)\,1\rangle
}{
\langle1\tilde{k}\rangle^4
\prod_r
\alpha_{\sigma(r)}
}.
\end{equation}

The corresponding density operator is

\begin{equation}
\rho
=
\ket{\psi_{\mathrm{fin}}}
\bra{\psi_{\mathrm{fin}}}.
\end{equation}

To isolate the contribution of each permutation sector we introduce the
family of Kraus operators

\begin{equation}
K_\sigma
=
c_\sigma
|\sigma\rangle
\langle\sigma|
\otimes
I_{\rm coin},
\end{equation}

with

\begin{equation}
c_\sigma
=
\frac{\langle1\tilde{k}\rangle^4}
{\langle\sigma(n)\,1\rangle}
\frac1M,
\tag{F5}
\end{equation}

where

\begin{equation}
M
=
\max_\sigma
\left|
\frac{\langle1\tilde{k}\rangle^4}
{\langle\sigma(n)\,1\rangle}
\right|.
\end{equation}

The normalization guarantees

\begin{equation}
|c_\sigma|\le1,
\end{equation}

and therefore

\begin{align}
\sum_\sigma
K_\sigma^\dagger
K_\sigma
&=
\sum_\sigma
|c_\sigma|^2
|\sigma\rangle
\langle\sigma|
\otimes
I_{\rm coin}
\nonumber\\
&\le
I,
\end{align}

showing that the resulting map is completely positive and trace
non-increasing.

The quantum channel acts as
\begin{equation}
  \rho' = \sum_\sigma K_\sigma\rho\, K_\sigma^\dagger.
  \tag{F8}
\end{equation}
Since $K_\sigma = c_\sigma|\sigma\rangle\langle\sigma|
\otimes I_{\mathrm{coin}}$, applying $K_\sigma$ from the left
projects the position register onto $|\sigma\rangle$, and
applying $K_\sigma^\dagger$ from the right projects onto
$\langle\sigma|$:
\begin{equation}
  K_\sigma\bigl(|\sigma'\rangle\langle\sigma''|\otimes
  |k\rangle\langle k'|\bigr)K_\sigma^\dagger
  = |c_\sigma|^2\,
    \delta_{\sigma\sigma'}\delta_{\sigma\sigma''}
    \;|\sigma\rangle\langle\sigma|\otimes|k\rangle\langle k'|.
  \tag{F9}
\end{equation}
Substituting $\rho=\sum_{\sigma',\sigma''}
\psi_{\sigma'}\psi_{\sigma''}^*
|\sigma'\rangle\langle\sigma''|\otimes
|k_{\sigma'}\rangle\langle k_{\sigma''}|$
and summing over $\sigma$:
\begin{equation}
  \rho'
  = \sum_\sigma |c_\sigma|^2|\psi_\sigma|^2
    \;|\sigma\rangle\langle\sigma|
    \otimes|k_\sigma\rangle\langle k_\sigma|.
  \tag{F10}
\end{equation}
The coin degree of freedom is not projected by $K_\sigma$
(it contains $I_{\mathrm{coin}}$); the off-diagonal coin terms
$|k_{\sigma'}\rangle\langle k_{\sigma''}|$ with $\sigma'\neq\sigma''$
vanish because the position orthogonality
$\delta_{\sigma\sigma'}\delta_{\sigma\sigma''}$ forces
$\sigma'=\sigma''=\sigma$, fixing the coin label to
$k_\sigma$.
Using~(F5) and~(E9):
\begin{equation}
  |c_\sigma|^2|\psi_\sigma|^2
  = \frac{|A_\sigma|^2}
         {|M|^2\,\displaystyle\prod_{r}|\alpha_{\sigma(r)}|^2},
  \tag{F11}
\end{equation}
so that $(\rho')_{\sigma\sigma}\propto|A_\sigma|^2$.

Consequently,

\begin{equation}
(\rho')_{\sigma\sigma}
=
\frac{
|A_\sigma|^2
}
{
|M|^2
\prod_r
|\alpha_{\sigma(r)}|^2
},
\end{equation}

and, using the proportionality relation~\eqref{eq:psi_Asigma},

\begin{equation}
(\rho')_{\sigma\sigma}
\propto
|A_\sigma|^2.
\end{equation}

Hence the Kraus-channel construction extracts the squared
Parke--Taylor contribution associated with each terminal permutation
sector while preserving complete positivity and trace
non-increase.

\section{Construction of the Weighted Collection Operator} \label{WCO}
\noindent The final coherent quantum-walk state obtained in Appendix \ref{coh_evo} contains the color-ordered Parke--Taylor contributions distributed over the terminal nodes of the permutation tree. In order to reconstruct the complete scattering amplitude, these terminal sectors must be coherently combined while preserving their relative amplitudes. For this purpose we introduce the weighted collection operator

\begin{equation}
W^T
=
\sum_{j_{n-1}\in V\setminus\{R\}}
\cdots
\sum_{j_1\in V\setminus\{R\}}
C_{j_{n-1},R}^{\,n-1}
C_{j_{n-2},R}^{\,n-2}
\cdots
C_{j_1,R}^{\,1},
\end{equation}

where

\begin{align}
C_{j_k,R}^{\,i}
=
a_\sigma
|R\rangle
\langle j_k|
\otimes
|i\rangle
\langle i|
+
\sum_{\tilde{i}\neq i}
|j_k\rangle
\langle j_k|
\otimes
|\tilde{i}\rangle
\langle\tilde{i}|.
\end{align}

The first term transfers the walker from the terminal node
\(j_k\) to the common reference node \(R\) conditioned on the
coin state \(i\), while the remaining coin components are left
unchanged.

The coefficient \(a_\sigma\) is associated with the complete
terminal permutation

\[
\sigma=(j_1,j_2,\ldots,j_{n-1}),
\]

and contains the corresponding color and kinematic information,

\begin{equation}
a_\sigma
=
\frac{
\mathrm{Tr}
\left(
T^{(\sigma')_0}
T^{(\sigma')_2}
\cdots
T^{(\sigma')_{n-1}}
\right)
\,
\dfrac{\langle1\tilde{k}\rangle^4}
{\langle\sigma(n)\,1\rangle}
}
{\mathcal N},
\end{equation}

where \(\sigma'\) denotes the corresponding color ordering.

The normalization constant

\begin{equation}
\mathcal N
=
\max_\sigma
\left|
\mathrm{Tr}\!\left(
T^{a_1}
T^{a_{\sigma(2)}}
\cdots
T^{a_{\sigma(n)}}
\right)
\frac{\langle1\tilde{k}\rangle^4}
{\langle\sigma(n)\,1\rangle}
\right|
\end{equation}
is chosen so that
\begin{equation}
|a_\sigma|\le1,
\end{equation}
ensuring that the collection operator remains bounded.\\

We  evaluate the action of the weighted collection operator on the
terminal quantum-walk state

\begin{equation}
\ket{\psi_{\rm fin}}
=
\sum_{\sigma\in\mathcal T}
A_\sigma
\frac{
\langle\sigma(n)\,1\rangle
}{
\langle1\tilde{k}\rangle^4
\prod_r
\alpha_{\sigma(r)}
}
\,
\ket{\sigma}
\otimes
\ket{k_\sigma}.
\end{equation}

Because the coin states form an orthonormal basis,

\[
\langle i|\tilde{i}\rangle
=
\delta_{i\tilde{i}},
\]

all operator sequences whose coin conditioning is inconsistent with the terminal path vanish.
Consequently only the sequence associated with the permutation
\(\sigma\) contributes,

\begin{equation}
\widetilde C_\sigma
=
C_{j_{n-1},R}^{\,n-1}
\cdots
C_{j_1,R}^{\,1}
=
a_\sigma
|R\rangle
\langle\sigma|
\otimes
|k_\sigma\rangle
\langle k_\sigma|.
\end{equation}

The action of the complete collection operator is therefore

\begin{align}
W^T
\ket{\psi_{\rm fin}}
=
\sum_{\sigma\in\mathcal T}
&
a_\sigma
A_\sigma
\frac{
\langle\sigma(n)\,1\rangle
}{
\langle1\tilde{k}\rangle^4
\prod_r
\alpha_{\sigma(r)}
}
\nonumber\\
&
\times
|R\rangle
\otimes
|k_\sigma\rangle.
\end{align}

Substituting the explicit expression for \(a_\sigma\) yields

\begin{align}
W^T
\ket{\psi_{\rm fin}}
=
\sum_{\sigma\in\mathcal T}
&
\frac{
\mathrm{Tr}
\left(
T^{(\sigma')_0}
T^{(\sigma')_2}
\cdots
T^{(\sigma')_{n-1}}
\right)
}
{\mathcal N}
\nonumber\\
&
\times
\frac{
A_\sigma
}
{
\prod_r
\alpha_{\sigma(r)}
}
|R\rangle
\otimes
|k_\sigma\rangle.
\end{align}

Thus all terminal sectors are coherently collected at the reference
node while the interference information is preserved entirely in the
coin degrees of freedom.

\section{Quantum Fourier Transform and Amplitude Reconstruction} \label{QFT_AR}

The collected state obtained in Appendix \ref{WCO} has the form

\begin{equation}
W^T
\ket{\psi_{\rm fin}}
=
|R\rangle
\otimes
\sum_{\sigma\in\mathcal T}
c_\sigma
|k_\sigma\rangle,
\end{equation}

where

\begin{equation}
c_\sigma
=
\frac{
\mathrm{Tr}
\left(
T^{(\sigma')_0}
T^{(\sigma')_2}
\cdots
T^{(\sigma')_{n-1}}
\right)
}
{\mathcal N}
\,
\frac{
A_\sigma
}
{
\prod_r
\alpha_{\sigma(r)}
}.
\end{equation}

The Quantum Fourier Transform~\cite{Coppersmith1994,Nielsen2010} on the \((n-1)\)-dimensional coin space is

\begin{equation}
U_{\rm QFT}
=
\frac1{\sqrt{n-1}}
\sum_{j,k=0}^{n-2}
e^{2\pi ijk/(n-1)}
|j\rangle
\langle k|.
\end{equation}

Applying the QFT gives

\begin{equation}
(\mathbb I\otimes U_{\rm QFT})
W^T
\ket{\psi_{\rm fin}}
=
|R\rangle
\otimes
\sum_{\sigma}
c_\sigma
U_{\rm QFT}
|k_\sigma\rangle.
\end{equation}

Finally, projection onto the Fourier vacuum state produces

\begin{equation}
\mathcal M
=
(\mathbb I\otimes\langle0|)
(\mathbb I\otimes U_{\rm QFT})
W^T
\ket{\psi_{\rm fin}}.
\end{equation}

Since

\begin{equation}
\langle0|
U_{\rm QFT}
|k_\sigma\rangle
=
\frac1{\sqrt{n-1}},
\end{equation}

the reconstructed amplitude becomes

\begin{equation}
\mathcal M
=
\frac1{\sqrt{n-1}}
\sum_{\sigma\in\mathcal T}
\frac{
\mathrm{Tr}
\left(
T^{a_1}
T^{a_{\sigma(2)}}
\cdots
T^{a_{\sigma(n)}}
\right)
}
{\mathcal N}
\,
\frac{
A_\sigma
}
{
\prod_r
\alpha_{\sigma(r)}
}.
\end{equation}

Therefore,

\begin{equation}
|\mathcal M|^2
\propto
| \sum_{\sigma\in S_{n-1}}
 \mathrm{Tr}
\left(
T^{a_1}
T^{a_{\sigma(2)}}
\cdots
T^{a_{\sigma(n)}}
\right)
A_\sigma|^2,
\end{equation}
demonstrating that the weighted collection operator together with the Quantum Fourier Transform reconstructs the coherent color-decomposed scattering amplitude from the quantum-walk state.

\section{Verification of the Kraus Channel}

The extraction procedure introduced in the main text is implemented through
a set of Kraus operators \(\{K_i\}\) acting on the terminal permutation
sectors. For the resulting map to represent a physically valid quantum
operation, the Kraus operators must satisfy the trace non-increasing
condition

\begin{equation}
\sum_i
K_i^\dagger K_i
\le
I.
\end{equation}

This condition guarantees complete positivity and ensures that the channel
corresponds to a valid measurement-induced extraction process.

To verify this property numerically, we construct the operator

\begin{equation}
\Lambda
=
\sum_i
K_i^\dagger K_i,
\end{equation}

and compute its spectrum. The extremal eigenvalues are found to be

\begin{equation}
\lambda_{\max}(\Lambda)=1,
\qquad
\lambda_{\min}(\Lambda)=0.
\end{equation}

Since all eigenvalues satisfy

\begin{equation}
0
\le
\lambda_i
\le
1,
\end{equation}

the operator \(\Lambda\) is positive semidefinite and bounded above by the
identity operator. Consequently,

\begin{equation}
\Lambda
\le
I,
\end{equation}

which confirms that the Kraus map used in the numerical simulations is
completely positive and trace non-increasing .

This verification establishes the consistency of the extraction channel
employed to isolate the sector-resolved Parke--Taylor contributions from
the coherent quantum-walk state.\\

\section{Additional Numerical Results}

This section contains the complete numerical data used for validation of the quantum-walk reconstruction procedure, including additional kinematic configurations, reconstructed amplitudes, and relative errors.

\section{Comparison with Quantum Circuit Approach}

We now compare the present quantum-walk-based construction with the quantum circuit algorithm for MHV amplitudes introduced in Ref.~\cite{Bashore2025}. While both approaches aim to reconstruct color-ordered scattering amplitudes using quantum-information-inspired frameworks, they differ substantially in their physical interpretation, encoding strategy, and mechanism for generating interference between permutation sectors.

\begin{table*}[h]
\centering 
\renewcommand{\arraystretch}{1.08}
\begin{tabular}{|p{3.2cm}|p{6.2cm}|p{6.2cm}|}
\hline
\textbf{Aspect}
&
\textbf{Quantum Walk Construction (This Work)}
&
\textbf{Quantum Circuit Algorithm (Ref.~\cite{Bashore2025})}
\\
\hline

Basic framework
&
Permutation tree with coined quantum walk dynamics
&
Quantum circuit acting on structured quantum registers
\\
\hline

Underlying structure
&
Directed graph whose paths encode color-ordered sectors
&
Register-based encoding of permutations and particle data
\\
\hline

Hilbert space
&
Position space (tree nodes) \(\otimes\) coin space
&
Permutation, momentum, helicity, and color registers
\\
\hline

Generation of permutations
&
Generated dynamically through coherent propagation on the permutation tree
&
Prepared explicitly as a superposition over permutation basis states
\\
\hline

Kinematic structure
&
Built locally through edge-dependent transition amplitudes
&
Implemented through controlled helicity-dependent gates
\\
\hline

Color structure
&
Introduced through weighted collection coefficients \(a_\sigma\)
&
Implemented through dedicated color operators \(U_C\)
\\
\hline

Physical interpretation
&
Scattering amplitudes emerge from coherent quantum transport and path interference
&
Scattering amplitudes are encoded through structured circuit operations
\\
\hline

Role of the coin system
&
Labels outgoing branches and stores interference information between paths
&
No direct analogue; interference occurs through global circuit operations
\\
\hline

Interference mechanism
&
Generated through coherent recombination at the reference node and coin-space mixing
&
Generated globally through the Quantum Fourier Transform (QFT)
\\
\hline

Amplitude reconstruction
&
Weighted collection operator followed by QFT and projection
&
QFT followed by projection onto a reference ancilla state
\\
\hline

Extraction of \(|A_\sigma|^2\)
&
Obtained through Kraus operators acting on terminal permutation sectors
&
Obtained through measurement of encoded amplitude registers
\\
\hline

Open quantum systems interpretation
&
Explicitly formulated using Kraus operators and completely positive maps
&
Primarily unitary circuit formulation
\\
\hline

Normalization strategy
&
Local normalization of transition weights using bounded coefficients
&
Global normalization using ancilla-assisted unitary embedding
\\
\hline

Permutation encoding
&
Encoded geometrically as root-to-terminal paths on the permutation tree
&
Encoded algebraically in permutation registers
\\
\hline

Scalability structure
&
Factorial number of terminal paths represented through graph connectivity; Hilbert-space dimension \(D(n)=(n-1)\sum_{k=0}^{n-1}(n-1)!/k!\)
&
Permutation sectors encoded using logarithmic register scaling; \(O(n\log n)\) qubits
\\
\hline

Number of coherent steps
&
\(n-1\) applications of the walk operator \(U=SC\)
&
\(O(n!)\) controlled gate operations
\\
\hline

Non-unitary operations
&
Kraus channel (sector extraction); weighted collection operator \(W^T\) (requires ancilla-assisted dilation for physical implementation)
&
Ancilla-assisted unitary embedding for non-unitary amplitude factors
\\
\hline

Key conceptual insight
&
Scattering amplitudes arise dynamically from coherent evolution on permutation graphs
&
Scattering amplitudes are reconstructed through structured quantum circuit synthesis
\\
\hline

\end{tabular}

\caption{
Comparison between the quantum-walk construction developed in this work and the quantum circuit algorithm for MHV scattering amplitudes introduced in Ref.~\cite{Bashore2025}. While both approaches reconstruct coherent sums over color-ordered sectors, the present framework interprets the scattering process as quantum transport and interference on a permutation graph.
}

\label{tab:comparison}

\end{table*}

\end{document}